\documentclass[11pt]{article}

\usepackage[margin=1in]{geometry}
\usepackage{amssymb}
\usepackage{amsthm}
\usepackage{aliascnt}
\usepackage{amsmath}
\usepackage{booktabs}
\usepackage{hyperref}
    \hypersetup{
        linktocpage=true,
        colorlinks=true,
        linkcolor=blue,
        citecolor=blue,
        urlcolor=blue,
    }
\usepackage{cleveref}
\usepackage{xcolor}
\usepackage[numbers]{natbib}
\theoremstyle{plain}
\newtheorem{theorem}{Theorem}[section]
\crefname{theorem}{theorem}{theorems}
\Crefname{theorem}{Theorem}{Theorems}

\newaliascnt{proposition}{theorem}
\newtheorem{proposition}[proposition]{Proposition}
\aliascntresetthe{proposition}
\crefname{proposition}{proposition}{propositions}
\Crefname{proposition}{Proposition}{Propositions}

\newaliascnt{lemma}{theorem}
\newtheorem{lemma}[lemma]{Lemma}
\aliascntresetthe{lemma}
\crefname{lemma}{lemma}{lemmas}
\Crefname{lemma}{Lemma}{Lemmas}

\newaliascnt{corollary}{theorem}
\newtheorem{corollary}[corollary]{Corollary}
\aliascntresetthe{corollary}
\crefname{corollary}{corollary}{corollaries}
\Crefname{corollary}{Corollary}{Corollaries}

\theoremstyle{definition}
\newaliascnt{definition}{theorem}
\newtheorem{definition}[definition]{Definition}
\aliascntresetthe{definition}
\crefname{definition}{definition}{definitions}
\Crefname{definition}{Definition}{Definitions}

\newaliascnt{example}{theorem}
\newtheorem{example}[example]{Example}
\aliascntresetthe{example}
\crefname{example}{example}{examples}
\Crefname{example}{Example}{Examples}

\theoremstyle{remark}
\newaliascnt{remark}{theorem}
\newtheorem{remark}[remark]{Remark}
\aliascntresetthe{remark}
\crefname{remark}{remark}{remarks}
\Crefname{remark}{Remark}{Remarks}

\newcommand{\R}{\mathbb R}
\newcommand{\E}{\mathbb E}
\newcommand{\N}{\mathbb N}
\newcommand{\cF}{\mathcal F}
\newcommand{\cD}{\mathcal D}
\newcommand{\cH}{\mathcal H}
\newcommand{\vareps}{\varepsilon}
\newcommand{\Cal}{\textup{Cal}}
\newcommand{\Ber}{\textup{Ber}}
\newcommand{\smCE}{\textup{smCE}}
\newcommand{\Err}{\operatorname{Err}}
\newcommand{\Adv}{\operatorname{Adv}}
\newcommand{\ind}{\mathbf{1}}

\allowdisplaybreaks[4]

\newcommand{\authorblock}[2]{\parbox[t]{0.43\textwidth}{\centering
    #1\\[-0.1em]
    {\small #2}
  }}

\title{Truthful Calibration Measures for Sequential Prediction}
\hypersetup{
  pdftitle={Truthful Calibration Measures for Sequential Prediction},
  pdfauthor={Anagha Gokul, Jason Hartline, Lunjia Hu, Jonathan Ullman, and Yifan Wu}
}

\author{\makebox[\dimexpr\textwidth-2\tabcolsep\relax][c]{\authorblock
      {Anagha Gokul\thanks{Khoury College of Computer Sciences. Email:
      \texttt{gokul.a@northeastern.edu}.}}
      {Northeastern University}
    \hfill
    \authorblock
      {Jason Hartline\thanks{Department of Computer Science. Email:
      \texttt{hartline@northwestern.edu}.}}
      {Northwestern University}
  }
  \\[1.5em]
  \makebox[\dimexpr\textwidth-2\tabcolsep\relax][c]{\authorblock
      {Lunjia Hu\thanks{Khoury College of Computer Sciences. Email:
      \texttt{lunjia@alumni.stanford.edu}.}}
      {Northeastern University}
    \hfill
    \authorblock
      {Jonathan Ullman\thanks{Khoury College of Computer Sciences. Email: \texttt{j.ullman@northeastern.edu}.
      Supported by NSF awards CNS-2232692 and CNS-2247484.}}
      {Northeastern University}
  }
  \\[1.5em]
  \makebox[\dimexpr\textwidth-2\tabcolsep\relax][c]{\authorblock
      {Yifan Wu\thanks{Microsoft Research New England. Email:
      \texttt{yifan.wu2357@gmail.com}.}}
      {Microsoft Research}
  }
}

\date{}

\begin{document}

\maketitle

\begin{abstract}
Calibration requires probabilistic reports to be conditionally unbiased and reliably interpretable as probabilities. A calibration measure assigns numerical error to miscalibrated reports. \citet{haghtalab2024truthfulness} proposed an approximately truthful calibration measure for online prediction, leaving open whether exact truthfulness is compatible with completeness and soundness.
We resolve this question negatively for sequential binary prediction: exact truthfulness is incompatible with completeness and soundness, even for independent outcomes. We then show that this impossibility is specific to exact truthfulness. On the positive side, we give two general reductions that transform a base calibration measure satisfying a uniform completeness condition into additively and multiplicatively approximately truthful measures, respectively. Applying the multiplicative reduction, for every $0<\vareps<1$ we construct a sound and complete calibration measure that is $1+\exp\left(-\Omega(T^{1-\vareps})\right)$-multiplicatively truthful. This improves the approximate truthfulness guarantee of \citet{haghtalab2024truthfulness}. We also prove a general transfer theorem showing that the uniform completeness of one calibration measure transfers to another when they are quantitatively related. Using this result, we show that every consistent calibration measure \cite{utc} satisfies uniform completeness and can serve as the base calibration measure to achieve $1 + o(1)$ multiplicative truthfulness.

\end{abstract}

\newpage
\tableofcontents
\newpage
\section{Introduction}
Probabilistic prediction is increasingly important in high-stakes decision making, where predictions are used to quantify uncertainty. \citet{dawid} proposed calibration as a desirable property of probabilistic reports. Calibration requires predictions to be conditionally unbiased, thus interpretable at face value. For example, if a weather predictor reports a $70\%$ chance of rain on many days, then rain should occur on $70\%$ of those days.
Formally, for every report value $p$, conditional on report $p$, the event should occur with frequency approximately $p$.

A \emph{calibration measure} assigns an error to a sequence of probability reports and outcomes to quantify the level of miscalibration.
Classical examples include the Expected Calibration Error (ECE) and its smooth variants \citep{guo2017calibration,smooth,blasioksmooth}. Concretely, given $T$ samples $(p_t,y_t)_{t\in [T]}$ of binary probability reports, ECE calculates the averaged report bias in the sample. For every report value $p$, write $n_p$ as the number of times that $p$ is reported, and $\hat p:=\frac1{n_p}\sum_{t:p_t=p}y_t$ as the empirical frequency conditional on report value $p$. The report bias is defined as the absolute distance $|\hat p-p|$ between the conditional empirical frequency $\hat p$ and the reported value $p$. ECE is the average bias weighted by the fraction $n_p/T$ of samples assigned report $p$:
\[
\operatorname{ECE}_T(p_{1:T},y)
:=
\sum_{p\in \mathcal R}
\frac{n_p}{T}
\left|\hat p-p\right|,
\]
where $\mathcal R$ is the set of values reported by the predictor.

However, ECE does not incentivize truthful reporting even when the predictor knows the ground-truth data-generating distribution. \Cref{example: intro non truthful of ece} shows that while a truthful predictor suffers $\Omega\left(1/\sqrt T\right)$ sampling error, a strategic predictor can achieve ECE$=0$ by misreporting.
\begin{example}[\protect\citep{haghtalab2024truthfulness}]
\label{example: intro non truthful of ece}
    Split the time horizon into three equal blocks.  In the first block, outcomes are independent fair coin flips.  In the second block, every outcome is $0$.  In the third block, every outcome is $1$.  The truthful predictor reports $1/2$ in the first block, $0$ in the second block, and $1$ in the third block.
    The deterministic blocks are then perfectly calibrated.  The only contribution to ECE comes from the $\Theta\left(1/\sqrt T\right)$ sampling error in the first block.

    A strategic predictor may look perfectly calibrated by using the later deterministic blocks to correct this sampling error.  It also reports $1/2$ throughout the first block.  After seeing that block, if the conditional empirical frequency at report $1/2$ is above $1/2$, the predictor reports $1/2$ on enough later zero outcomes to balance them, and similarly otherwise.  At the end, the report $1/2$ has exactly $1/2$ of ones.  Thus the strategic predictor obtains ECE$=0$.
\end{example}

A recent line of work studies whether calibration measures can be made truthful. \citet{haghtalab2024truthfulness} initiated this direction in the sequential prediction setting by proposing an approximately truthful calibration measure. In the batch setting, \citet{hartline2026perfectly} and \citet{lu2026truthful} show that perfectly truthful calibration measures exist. These results leave a natural gap: whether exact truthfulness is possible in the sequential prediction problem.

Our paper answers this question negatively. We show that exact truthfulness is incompatible with the completeness and soundness requirements of a calibration measure in the sequential prediction problem. On the positive side, we reduce the construction of an approximately truthful calibration measure to finding a sound and complete base calibration measure that satisfies uniform completeness.

\subsection{Our Contributions}
Our first result shows that exact truthfulness is impossible for sequential calibration. The main idea is that completeness and soundness require a calibration measure to  assign lower expected error to calibrated product samples than to miscalibrated ones, when the sample size is sufficiently large. In the sequential setting, that comparison can occur on a rare history. With exponentially small but nonzero probability, the truthful transcript can look very miscalibrated, while an alternative transcript induced by a strategic report looks more calibrated. When this happens, the strategic predictor can lower expected calibration error by deviating to the alternative report. In \Cref{sec:product-dist-impossibility}, we formalize this intuition using a hybrid argument. Recall that completeness and soundness imply the separation of calibrated distribution from a miscalibrated one, when there are $T$ i.i.d.\ prediction outcome pairs from these distributions. A hybrid is a mixture where the first $m$ pairs are drawn from a miscalibrated distribution, and the later $T-m$ pairs are from a calibrated distribution.
This hybrid argument finds a local profitable deviation  at sample $m$ under some realization of prefix, by enumerating $m$ over all such hybrids. 

Our second result gives a quantitative lower bound for the gap from being truthful, which is exponentially small in $T$. In \Cref{sec:quantitative-bounds}, we lower bound the largest online gain from strategic misreporting by an exponentially small multiple of the offline distinguishing gap between the calibrated and miscalibrated distributions. Concretely, for the pair of distributions indexed by $q\in(0,1)$, the gain is at least $\frac{(1-q)^{T-1}}{qT}$ times the distinguishing gap (\Cref{thm:quantitative-hybrid-impossibility}); consequently, for every complete and sound measure and every $\rho\in(0,1)$, the gain is at least $c_\rho\rho^T$ for all sufficiently large $T$ (\Cref{cor:complete-sound-quantitative-advantage}). The exponential factor reflects the probability of reaching the rare prefix realization on which the local strategic deviation is useful.

Our third result shows that the impossibility is specific to exact truthfulness. In \Cref{sec:approximate-truthfulness}, we reduce the construction of an approximately truthful calibration measure to finding a sound and complete base calibration measure that satisfies uniform completeness. Here, uniform completeness means that, under truthful reporting and after subtracting a vanishing sampling-error threshold $\tau_T$, the expected error is at most a vanishing multiple $\eta_T$ of the average intrinsic variance of the outcomes, uniformly over all outcome distributions (\Cref{def:uniform-completeness}). The reduction proceeds in two stages. Subtracting the threshold $\tau_T$ first gives an $\eta_T/4$-additively truthful measure. Adding a squared-loss term with vanishing weight $\lambda_T$ then strengthens the guarantee to $(1+\eta_T/\lambda_T)$-multiplicative truthfulness. Both stages preserve the completeness and soundness of the base measure. Applying the multiplicative reduction with smooth calibration error as the base measure, using the concentration bound in \Cref{sec:smCE-conc}, gives for every $0<\vareps<1$ a complete and sound calibration measure that is $\left(1+\exp\left(-\Omega(T^{1-\vareps})\right)\right)$-multiplicatively truthful (\Cref{cor:smCE-truthfulness}), improving the $O(1)$ multiplicative factor of \citet{haghtalab2024truthfulness}.

Our fourth result transfers uniform completeness between calibration
measures using a pointwise comparison (\Cref{sec:more-error-bar-bounds}).
The source and target can be arbitrary calibration measures. Combined
with the kernel bound in \Cref{subsec:more-errors-kernel}, the transfer
establishes uniform completeness for every calibration measure
with a polynomial robust-completeness bound. We then apply the transfer theorem to the lower and upper distances to calibration, as well as interval calibration error.

\subsection{Related Work}

Existing literature proposes several desiderata of calibration measures, including truthfulness and continuity.

\paragraph{Sequential calibration.}
The classical literature on sequential forecasting established asymptotically calibrated procedures, including against adversarial outcome sequences \citep{foster1998asymptotic,foster1997calibrated,cal-minimax}. Subsequent work has investigated finite-horizon rates, developing stronger lower bounds for calibration error \citep{sidestep} and studying the related distance-to-calibration objective \citep{qiao-distance,arunachaleswaran2025elementary}. Most recently, \citet{dagan2024breakingt23barriersequential} broke the longstanding \(T^{2/3}\) barrier for sequential calibration.

\paragraph{Truthfulness.}
We discuss truthfulness first. 
Truthfulness was introduced for calibration measures by \citet{haghtalab2024truthfulness}. Recent literature observes that existing calibration measures can reward strategic misreporting (e.g., \citep{foster2021hedging,sidestep}). \citet{haghtalab2024truthfulness} introduce a calibration measure for the sequential prediction problem that is $O(1)$-multiplicatively truthful.
\citet{qiao2025truthfulness} show that exact truthfulness is incompatible with decision-theoretic calibration guarantees. In the batch setting, \citet{hartline2026perfectly} show that perfect truthfulness is compatible with soundness and completeness, and \citet{lu2026truthful} extend this direction to multi-class prediction. We consider the sequential prediction problem. Our results identify the incompatibility of exact truthfulness with the completeness and soundness of calibration measures, and improve the approximation in \citet{haghtalab2024truthfulness}.

\paragraph{Continuity.}
Continuity is another natural desideratum. Continuity requires that a calibration measure should be robust to small perturbations of the reported probabilities, with the outcomes held fixed. This rules out ECE and binned ECE, where an empirically calibrated sequence can acquire large error after arbitrarily small perturbations of the reports because the perturbation splits previously identical reports into separate buckets. The distance-to-calibration framework of \citet{utc} and smooth calibration error \citep{smooth} propose continuous calibration measures.

\paragraph{Independent work on impossibility of truthfulness.} 
Independent of our work, two recent manuscripts by Dewasurendra establish related impossibility results for perfectly truthful, sound, and complete sequential calibration measures \citep{Dewasurendra2026characterizingperfect,Dewasurendra2026noperfectly}. Their notions of truthfulness, soundness, and completeness differ subtly from those in our main impossibility result in \Cref{sec:product-dist-impossibility}. In particular, they require the stronger notion of truthfulness that allows the outcomes $y_1,\ldots,y_T$ to follow an arbitrary joint distribution on $\{0,1\}^T$, whereas our main result requires truthfulness only for product outcome distributions. Their soundness and completeness notions are aligned with the setting of \citet{haghtalab2024truthfulness}. Accordingly, their results are closer to our impossibility result in \Cref{sec:weak-soundness-appendix}.

\section{Technical Overview}
\label{sec:technical-overview}

Before giving the formal definitions, we summarize the proof ideas behind
the three main parts of the paper: the exact impossibility, the
quantitative lower bounds obtained from that proof, and the approximate
construction.  The notation used here is made formal in
\Cref{sec:preliminaries}; in particular, $\delta_z$ denotes the point mass
at $z$ and $x_+=\max\{x,0\}$.

\subsection{Impossibility of Truthfulness.}
The source of the impossibility is that exact truthfulness asks the predictor to report
the next conditional mean even on rare histories where the transcript so
far looks badly miscalibrated. Completeness and soundness, however,
reward transcripts that look calibrated. On such histories, a strategic
predictor may prefer a report that repairs the accumulated imbalance.

\paragraph{Warm-Up Example: Symmetric Case.} We first illustrate the
idea under a symmetry assumption: suppose the calibration measure depends
only on the empirical distribution of the transcript, not on the order of
the time steps.  Let $T=2n$.  In the first $n$ rounds the ground-truth
data-generating distribution is Bernoulli with mean $1/2$, so the
truthful report is $1/2$.  In the remaining $n$ rounds the
outcome is deterministically $1$, so the truthful report is $1$.  Now condition
on the rare history where every outcome in the first block is realized to be $0$, which happens with probability $2^{-n}$.
On that history, the truthful history of prediction-outcome pairs $(r, y)$ contains $n$ copies of
$(1/2,0)$ followed by $n$ copies of $(1,1)$.

A strategic predictor can behave truthfully everywhere except on this
rare history.  If the first block is all zeros, it keeps reporting $1/2$
during the deterministic-one rounds instead of switching to $1$.  Then the
transcript contains $n$ copies of $(1/2,0)$ and $n$ copies of
$(1/2,1)$, which looks perfectly calibrated at report $1/2$. Completeness
and soundness then favor the calibrated history produced by the strategic
predictor. Thus, the strategic predictor has a lower error on this rare
history and achieves the same error as the truthful predictor everywhere else.  Truthful
reporting is then not an optimal strategy.

\paragraph{Hybrid Argument: General Case.} The full result removes the
symmetry assumption of the calibration measure using a hybrid argument.
We start with two offline one-sample distributions over
prediction-outcome pairs, one calibrated and one miscalibrated.
Completeness and soundness imply that, for all sufficiently large $T$,
their $T$-fold product distributions are separated in expectation.  The
hybrid argument then changes the transcript one coordinate at a time, so
some adjacent pair of hybrids must carry a positive part of the endpoint
gap. That coordinate becomes the time step where the strategic online
predictor benefits from a deviating report.

For the overview, take the fair-coin pair
\[
\cD=(1/2)\delta_{(1/2,0)}+(1/2)\delta_{(1/2,1)}
\qquad\text{and}\qquad
\cD'=(1/2)\delta_{(1/2,0)}+(1/2)\delta_{(1,1)}.
\]
The distribution $\cD$ is calibrated. The distribution $\cD'$ is miscalibrated because the report $1/2$ appears only with outcome $0$.

For all sufficiently large $T$, the calibration measure gives smaller expected error to $\cD^T$ than to $(\cD')^T$. To connect these two histories, we define $T+1$ hybrids, which are random sequences of prediction-outcome pairs. Hybrid $j\in\{0,\ldots,T\}$ is constructed such that:
\begin{itemize}
\item the first $j$ samples are drawn from the miscalibrated distribution $\cD'$;
\item the later $T-j$ samples are drawn from the calibrated distribution $\cD$.
\end{itemize}
The first hybrid, with switching point $0$, is fully calibrated, and the
last hybrid, with switching point $T$, is fully miscalibrated. Hence some
adjacent pair of hybrids must account for a positive part of the endpoint
gap. These two hybrids share the same prefix before time $m$, where each
pair is either $(1/2,0)$ or $(1,1)$, and the same calibrated
Bernoulli$(1/2)$ tail after time $m$. At time $m$, they agree on the
zero-outcome pair $(1/2,0)$ and differ only when the outcome is $1$: the
calibrated version reports $1/2$ with $(1/2,1)$, while the miscalibrated
version reports $1$ with $(1,1)$. Since this adjacent gap is positive on
average over the common prefix, there is a fixed prefix realization of
length $m-1$ for which the same local preference holds. This fixed prefix
is the rare history where the predictor benefits from misreporting:
\begin{itemize}
\item it is generated by $m-1$ i.i.d.\ draws from \(\cD'\), so the reports are either $1/2$ or $1$, which specifies the construction of the data-generating distribution;
\item conditioning on the prefix, if the \(m\)th outcome is deterministically \(1\) and the later rounds are drawn from the calibrated Bernoulli\((1/2)\) tail, the expected calibration error is lower when time \(m\) reports \(1/2\) than when it reports the truthful value \(1\).
\end{itemize}

\paragraph{Lower Bounds for Approximate Truthfulness.}
The same proof gives a quantitative statement.  Let
$\Delta_T(1/2)$ be the offline distinguishing gap between $(\cD')^T$ and
$\cD^T$, and let $\operatorname{Adv}^{\mathrm{prod}}_T(\Cal)$ be the
largest additive gain achievable by a non-anticipating misreporting
strategy under a product data-generating distribution.  In the fair-coin case, the hybrid
proof yields
\[
\operatorname{Adv}^{\mathrm{prod}}_T(\Cal)
\ge
\frac{2^{2-T}}{T}\Delta_T(1/2).
\]
The factor is exponentially small because the online strategy must wait
for the particular prefix selected by the hybrid argument.

The appendix generalizes the impossibility of truthfulness under the weak
completeness and soundness requirements in
\citet{haghtalab2024truthfulness}. That work requires a calibration measure to distinguish the calibration from miscalibration of only constant predictors. This weak completeness
asks that a constant predictor who reports the true Bernoulli mean
has small expected error, while soundness only asks that a constant
report separated from that mean has large expected error. Our main
definitions are stronger: they require the measure to distinguish
calibrated from miscalibrated one-sample distributions for arbitrary
predictors. A comparison lemma shows
that, under truthfulness, the expected error of a constant misreport $q$
under $\Ber(p)^T$ is at most $p(1-q)/(q(1-p))$ times the error of the
truthful constant report $p$.  This bound is independent of $T$, and
therefore shows the impossibility of truthfulness under the weak completeness and soundness requirement.

\subsection{Approximate truthfulness.}
We reduce the construction of an approximately truthful calibration measure to finding a sound and complete base calibration measure that satisfies uniform completeness.

\paragraph{Additively Approximate Truthfulness.} Given such a base measure, the additive truthfulness construction subtracts a sampling-error threshold $\tau_T$:
\[
\Cal_T^{\mathrm{add}}(r,y)
=
\bigl(\Err_T(r,y)-\tau_T\bigr)_+.
\]
We require the base measure $\Err$ to satisfy uniform completeness: for any ground truth distribution, the tail after extracting a sampling error $\tau_T$ can be bounded. Additive approximate truthfulness then follows immediately because the measure is non-negative: $\Cal_T^{\mathrm{add}}$ is $\eta_T/4$-additively truthful, since the average intrinsic variance is at most $1/4$.

\paragraph{Multiplicatively Approximate Truthfulness.} The multiplicative approximation adds a small normalized squared-loss term to the additive one:
\[
\Cal_T^{\mathrm{mult}}(r,y)
=
\Cal_T^{\mathrm{add}}(r,y)
+
\lambda_T\cdot\frac1T\sum_{t=1}^T(r_t-y_t)^2.
\]
The squared-loss term is proper round by round \citep{gneiting2007strictly} and gives a
universal lower bound on the expected error of every strategy at the
scale of the intrinsic variance
$\sum_t\E[p_t^\ast(1-p_t^\ast)]/T$.  Consequently,
$\Cal_T^{\mathrm{mult}}$ is
$1+\eta_T/\lambda_T$ multiplicatively truthful.
For smooth calibration error, a martingale concentration bound gives
$\tau_T=T^{-\alpha}$ and $\eta_T=\exp(-T^{1-2\alpha}/32)$ for every
$0<\alpha<1/2$ and all sufficiently large $T$. Taking
$\lambda_T=T^{-\alpha}$ yields the factor
$1+\exp(-T^{1-2\alpha}/64)$, which is the
$1+\exp(-\Omega(T^{1-\vareps}))$ guarantee with $\alpha=\vareps/2$.
Bounded kernel calibration error also satisfies uniform completeness.
The pointwise transfer in \Cref{sec:more-error-bar-bounds} then extends
the construction to other calibration measures, including lower and
upper distance to calibration and interval calibration error.

\section{Preliminaries}
\label{sec:preliminaries}

We study binary sequential prediction over a fixed horizon $T\in\N$.
Write $[T]=\{1,\ldots,T\}$ and let $\Ber(p)$ denote the Bernoulli
distribution with mean $p$.  A realization is denoted
$Y=(Y_1,\ldots,Y_T)\in\{0,1\}^T$.  We write $\delta_z$ for the point mass
at $z$ and $x_+=\max\{x,0\}$ for the positive part of a real number.

\subsection{Sequential Prediction}

A data-generating distribution is a joint distribution over the outcome
sequence \(Y=(Y_1,\ldots,Y_T)\).  

At each time $t$, a predictor observes the past outcomes
$Y_{<t}=(Y_1,\ldots,Y_{t-1})$ and then reports a probability
$r_t\in[0,1]$.  The predictor is assumed to know the ground-truth
data-generating distribution before play begins.  We use $r$ for an
arbitrary report sequence and $p^\ast$ for the truthful sequence.

The first object we need is the set of admissible reporting rules.  The
non-anticipation restriction below is what makes the problem sequential:
reports may depend on the past and on the known ground truth, but not on
future realized outcomes.

\begin{definition}[Prediction strategy]
A prediction strategy, relative to a fixed known ground-truth
data-generating distribution, is a sequence of maps
$S_t:\{0,1\}^{t-1}\to[0,1]$.  Given an outcome sequence $Y$, it produces
the predictable report vector
\[
S(Y)=\bigl(S_1,S_2(Y_1),\ldots,S_T(Y_{<T})\bigr).
\]
The maps \(S_t\) may depend on this data-generating distribution, but
not on future realized outcomes.  We suppress this dependence from the
notation.
\end{definition}

Two classes of data-generating distributions are used throughout.
In the \emph{product-distribution} case, there are numbers
$p_1^\ast,\ldots,p_T^\ast\in[0,1]$ and the outcomes are generated
independently as $Y_t\sim\Ber(p_t^\ast)$.  In the \emph{correlated}
case, the ground truth is an arbitrary joint
distribution $\Pi$ on $\{0,1\}^T$, and the truthful report is the
conditional mean
\[
p_t^\ast(Y_{<t})=\E_\Pi[Y_t\mid Y_{<t}].
\]
Thus $p_t^\ast$ is predictable with respect to the outcome history.

\subsection{Calibration Measures, Completeness and Soundness}

A calibration measure assigns
a non-negative error to a transcript of probability reports and realized
outcomes; all truthfulness and soundness requirements are constraints on
this map.

\begin{definition}[Calibration measure]
A calibration measure is a map
\[
\Cal_T:[0,1]^T\times\{0,1\}^T\to\R_{\ge 0}.
\]
Smaller values of $\Cal_T$ represent better calibration.
\end{definition}

When $Z_t=(r_t,y_t)\in[0,1]\times\{0,1\}$ and
$Z=(Z_1,\ldots,Z_T)$, we also write
\[
\Cal_T(Z)=\Cal_T(r_1,\ldots,r_T;y_1,\ldots,y_T).
\]
For a distribution $\cD$ over prediction-outcome pairs, $\cD^T$ denotes
the distribution of $T$ independent draws from $\cD$.  A sequence
distribution is a probability distribution $\Lambda$ on
$([0,1]\times\{0,1\})^T$.

Completeness and soundness distinguish calibrated data from
miscalibrated data.  The following definition fixes what calibrated means
for a one-sample distribution over a report and an outcome.

\begin{definition}[Calibrated distribution]
A distribution $\cD$ over $[0,1]\times\{0,1\}$ is calibrated if, for
$(R,Y)\sim\cD$,
\[
\E[Y\mid R]=R
\qquad\text{almost surely}.
\]
\end{definition}

Completeness and soundness formalize the statistical meaning of a
calibration measure. Following \citet{hartline2026perfectly}, we use the
sample-access asymptotic definitions: calibrated product distributions
should have vanishing expected error, while miscalibrated product
distributions should have error bounded away from zero in the limit.
Since our measure $\Cal_T$ is written directly as a function of a
length-$T$ transcript, evaluating it on $Z\sim\cD^T$ is the same
sample-access experiment.

Completeness is the vanishing-error side of the statistical guarantee:
if the reports are genuinely calibrated, then the expected calibration
error should disappear as the sample size grows.

\begin{definition}[Completeness]
\label{def:completeness}
A sequence of calibration measures
$\Cal=(\Cal_T)_{T\ge1}$ is complete if, for every calibrated distribution
$\cD$ over $[0,1]\times\{0,1\}$,
\[
\lim_{T\to\infty}
\E_{Z\sim\cD^T}[\Cal_T(Z)]
=0.
\]
\end{definition}

Soundness is the separating side of the same guarantee: a fixed
miscalibrated distribution should retain a positive expected error in
the large-sample limit.

\begin{definition}[Soundness]
\label{def:soundness}
A sequence of calibration measures
$\Cal=(\Cal_T)_{T\ge1}$ is sound if, for every miscalibrated distribution
$\cD$ over $[0,1]\times\{0,1\}$,
\[
\liminf_{T\to\infty}
\E_{Z\sim\cD^T}[\Cal_T(Z)]
>0.
\]
\end{definition}

Together, completeness and soundness imply that every fixed calibrated
one-sample distribution is eventually separated in expected error from
every fixed miscalibrated one-sample distribution.  The appendix also
discusses a weaker constant-predictor notion from
\citet{haghtalab2024truthfulness}, where one only asks calibrated
constant reports to have small expected error and separated constant
reports to have larger expected error.  The appendix gives a quantitative
obstruction for this weaker notion: truthfulness forces the former error
to be at least a constant fraction of the latter; see
\Cref{sec:weak-soundness-appendix}.

The normalized smooth calibration error \citep{smooth, utc} used later is
\begin{equation}
    \label{eq: smooth def}
\smCE_T(r,y)
:=
\frac1T\sup_{f\in\cF}
\left|
\sum_{t=1}^T f(r_t)(y_t-r_t)
\right|,
\end{equation}
where $\cF$ is the class of $1$-Lipschitz functions
$f:[0,1]\to[-1,1]$.

\subsection{Truthfulness}

We first define the truthful strategy, which reports the conditional
probability of the next outcome under the ground truth.  

\begin{definition}[Truthful strategy]
Given a ground-truth data-generating distribution \(\Pi\), a prediction
strategy \(S^\ast\) is truthful if, for every \(t\) and every history
\(y_{<t}\) with \(\Pr_\Pi[Y_{<t}=y_{<t}]>0\),
\[
S_t^\ast(y_{<t})
=
\E_\Pi[Y_t\mid Y_{<t}=y_{<t}].
\]
Values on histories with zero probability are immaterial.  We write
\(p^\ast=S^\ast(Y)\) for the realized truthful report sequence.
\end{definition}

Truthfulness of a calibration measure \citep{haghtalab2024truthfulness} requires the truthful strategy to minimize expected calibration
error for every joint data-generating distribution. 

\begin{definition}[Truthfulness]
A calibration measure $\Cal_T$ is truthful if, for every joint
distribution $\Pi$ on $\{0,1\}^T$, with truthful reports
\[
p_t^\ast(Y_{<t})=\E_\Pi[Y_t\mid Y_{<t}],
\]
and every prediction strategy $S$ that may depend on $\Pi$, we have
\[
\E_\Pi[\Cal_T(p^\ast,Y)]
\le
\E_\Pi[\Cal_T(S(Y),Y)].
\]
\end{definition}

We also use approximate versions of this truthfulness requirement.  The
additive version is the scale used by the quantitative lower bound in
\Cref{sec:quantitative-bounds}.

\begin{definition}[Additive approximate truthfulness]
\label{def:additive-truthfulness}
A calibration measure $\Cal_T$ is $\beta_T$-additively truthful for
product distributions if, for every product distribution and every
prediction strategy $S$ that may depend on that product distribution,
\[
\E[\Cal_T(p^\ast,Y)]
\le
\E[\Cal_T(S(Y),Y)]+\beta_T.
\]
The version for arbitrary joint distributions requires the same
inequality for every joint distribution of $Y$ and every prediction
strategy that may depend on that joint distribution.
\end{definition}

The second stage of the positive construction gives the stronger
multiplicative guarantee.  We define that version separately because both
notions are used in \Cref{sec:approximate-truthfulness}.

\begin{definition}[Multiplicative truthfulness]
\label{def:multiplicative-truthfulness}
A calibration measure $\Cal_T$ is $\alpha_T$-multiplicatively truthful
for product distributions if, for every product distribution and every
prediction strategy $S$ that may depend on that product distribution,
\[
\E[\Cal_T(p^\ast,Y)]
\le
\alpha_T\E[\Cal_T(S(Y),Y)].
\]
The version for arbitrary joint distributions requires the same
inequality for every joint distribution of $Y$ and every prediction
strategy that may depend on that joint distribution.
\end{definition}

\section{Product-Distribution Impossibility}
\label{sec:product-dist-impossibility}

We first prove that exact truthful reporting can fail even under product
data-generating distributions. Even a very weak finite-sample
discrimination requirement is enough: if a calibration measure assigns
smaller expected error to an i.i.d.\ calibrated distribution than to a
simple i.i.d.\ miscalibrated distribution, then the measure creates an
incentive for a strategic predictor to misreport. The proof is a hybrid
argument. The offline separation between the two distributions can be
localized to a single round where the truthful report is $1,$ but the
strategic predictor prefers the report $q.$ We then use this local
preference to construct an online prediction problem with a product
data-generating distribution. \\

The two product distributions in the theorem below represent the weakest
natural test that distinguishes a calibrated report from a miscalibrated
one at a single report value. Under
\[\cD_q=(1-q)\delta_{(q, 0)}+q\delta_{(q, 1)},\] the predictor always reports $q,$ and the outcome has conditional mean $q,$ so this distribution is calibrated. In contrast, under
\[
\cD'_q=(1-q)\delta_{(q,0)}+q\delta_{(1,1)},\] the report $q$ appears only when the outcome is $0,$ so the conditional mean of the outcome, given report $q$, is $0$ rather than $q.$ Thus $\cD'_q$ is miscalibrated in this elementary way.

The finite-sample theorem below is the core of the impossibility.  Any
measure that distinguishes this calibrated pair from the miscalibrated
pair already creates a profitable online deviation, even under
independent outcomes.

\begin{theorem}[Hybrid impossibility]
\label{thm:hybrid-impossibility}
Fix $T\in\N$ and $q\in(0,1)$.  Let
\[
\cD_q=(1-q)\delta_{(q,0)}+q\delta_{(q,1)}
\qquad\text{and}\qquad
\cD'_q=(1-q)\delta_{(q,0)}+q\delta_{(1,1)}.
\]
If a calibration measure $\Cal_T$ satisfies
\begin{equation}
\label{eq:hybrid-distinguish}
\E_{Z\sim\cD_q^T}[\Cal_T(Z)]
<
\E_{Z\sim(\cD'_q)^T}[\Cal_T(Z)],
\end{equation}
then there is a product data-generating distribution and a prediction
strategy \(S\) such that, writing \(p^\ast\) for the truthful report
sequence,
\[
\E[\Cal_T(S(Y),Y)]
<
\E[\Cal_T(p^\ast,Y)].
\]
\end{theorem}

\begin{proof}
For $m=0,1,\ldots,T$, define a hybrid distribution $\cH_m$ over
$T$ prediction-outcome pairs by drawing the first $m$ pairs independently
from $\cD'_q$ and the remaining $T-m$ pairs independently from $\cD_q$.
Thus $\cH_0=\cD_q^T$ and $\cH_T=(\cD'_q)^T$.  By
\eqref{eq:hybrid-distinguish}, there is some $m\in[T]$ such that
\[
\E_{Z\sim\cH_{m-1}}[\Cal_T(Z)]
<
\E_{Z\sim\cH_m}[\Cal_T(Z)].
\]

The two hybrids differ only in coordinate $m$.  In that coordinate,
$\cD_q$ and $\cD'_q$ agree on $(q,0)$ and differ only on the outcome-one
point:
\[
(q,1)\quad\text{under }\cD_q,
\qquad
(1,1)\quad\text{under }\cD'_q.
\]
Let
\[
Z_{<m}=((P_1,Y_1),\ldots,(P_{m-1},Y_{m-1}))
\]
denote the common prefix random vector of $\cH_{m-1}$ and $\cH_m$.  Its
support is $\{(q,0),(1,1)\}^{m-1}$. 
For a prefix
\[
z_{<m}=((p_1,y_1),\ldots,(p_{m-1},y_{m-1}))
\] and $a \in \{q, 1\},$ define 
\[\Phi_a(z_{<m}) := \E_{Y_{m+1:T}\sim \Ber(q)^{\otimes (T-m)}}
\Big[
\Cal_T(p_1,\ldots,p_{m-1},a,q,\ldots,q;
y_1,\ldots,y_{m-1},1,Y_{m+1:T})
\Big].\]

\begin{equation}
\label{eq:hybrid-difference}
\E_{Z\sim\cH_m}[\Cal_T(Z)]
-
\E_{Z\sim\cH_{m-1}}[\Cal_T(Z)]
=
q\,
\E_{Z_{<m}}
\bigl[
\Phi_1(Z_{<m})-\Phi_q(Z_{<m})
\bigr].
\end{equation}
Here the $(q,0)$ contribution cancels, since it appears with the same
probability $1-q$ under both hybrids.

Because the left-hand side of \eqref{eq:hybrid-difference} is strictly
positive, there exists a prefix
$\hat z_{<m}$
in the support of $Z_{<m}$ such that
\begin{equation}
\label{eq:hybrid-local-preference}
    \Phi_q(\hat z_{<m})<\Phi_1(\hat z_{<m})
\end{equation}

Equivalently, after this prefix and conditioned on $Y_m=1,$ reporting $q$ at time $m$ gives strictly smaller expected calibration error than reporting $1.$

We now build a prediction problem with a product data-generating
distribution.  Let the truthful
probabilities be
\[
p_t^\ast=
\begin{cases}
\hat p_t, & t<m,\\
1, & t=m,\\
q, & t>m,
\end{cases}
\]
and draw the outcomes independently as $Y_t\sim\Ber(p_t^\ast)$.

Define a prediction strategy $S=(S_1,\ldots,S_T)$ by
\[
S_t(y_{<t})=
\begin{cases}
p_t^\ast, & t\neq m,\\
q, & t=m \text{ and } y_{<m}=(\hat y_1,\ldots,\hat y_{m-1}),\\
1, & t=m \text{ and } y_{<m}\neq(\hat y_1,\ldots,\hat y_{m-1}).
\end{cases}
\]
Equivalently, $S$ reports truthfully on all histories except the history
\[
E=\{Y_1=\hat y_1,\ldots,Y_{m-1}=\hat y_{m-1}\}.
\]
On $E$, at time $m$ it reports $q$ instead of the truthful value $1$.

The event $E$ has positive probability, because
\[
\Pr(E)=\prod_{t=1}^{m-1}\Pr(Y_t=\hat y_t)
=
\prod_{t:\,\hat p_t=q}(1-q)
>0;
\]
when $\hat p_t=1$, necessarily $\hat y_t=1$ and the corresponding factor
is $1$.  Conditioned on $E$, the outcome at time $m$ is equal to $1$
almost surely and the remaining outcomes are independent with distribution
$\Ber(q)$.  Therefore \eqref{eq:hybrid-local-preference} is exactly the
statement that
\[
\E[\Cal_T(S(Y),Y)\mid E]
<
\E[\Cal_T(p^\ast,Y)\mid E].
\]
On $E^c$, the two report sequences coincide, so
\[
\E[\Cal_T(S(Y),Y)\mid E^c]
=
\E[\Cal_T(p^\ast,Y)\mid E^c].
\]
Since $\Pr(E)>0$, taking total expectations gives
\[
\E[\Cal_T(S(Y),Y)]
<
\E[\Cal_T(p^\ast,Y)].
\]
Thus, for this product distribution, the strategic predictor achieves
strictly smaller expected calibration error than the truthful predictor.
\end{proof}

The asymptotic completeness and soundness definitions eventually imply
the finite-sample separation required by the hybrid theorem.  The corollary contains the resulting no-deviation impossibility result.

\begin{corollary}[Completeness and soundness rule out truthfulness]
\label{cor:complete-sound-not-truthful}
No sequence of calibration measures
$\Cal=(\Cal_T)_{T\ge1}$ can be complete and sound while also satisfying
the following no-deviation condition at every horizon:
\[
\E[\Cal_T(p^\ast,Y)]
\le
\E[\Cal_T(S(Y),Y)]
\]
for every product data-generating distribution and every prediction
strategy \(S\), where \(p^\ast\) is the truthful report sequence.
\end{corollary}

\begin{proof}
Fix any $q\in(0,1)$ and let $\cD_q,\cD'_q$ be the pair of distributions
from \Cref{thm:hybrid-impossibility}.  The distribution $\cD_q$ is
calibrated and $\cD'_q$ is miscalibrated.  By completeness,
\[
\lim_{T\to\infty}
\E_{Z\sim\cD_q^T}[\Cal_T(Z)]
=0,
\]
while by soundness there is a constant $\gamma>0$ such that, for all
sufficiently large $T$,
\[
\E_{Z\sim(\cD'_q)^T}[\Cal_T(Z)]
\ge \gamma.
\]
For all sufficiently large $T$, the calibrated expectation is smaller
than $\gamma$, so
\[
\E_{Z\sim\cD_q^T}[\Cal_T(Z)]
<
\E_{Z\sim(\cD'_q)^T}[\Cal_T(Z)].
\]
\Cref{thm:hybrid-impossibility} then implies that $\Cal_T$ is not
minimized by truthful reporting on some product data-generating
distribution at that horizon.
\end{proof}

We note that the theorem is deliberately finite-sample. It does not
require any asymptotic notion of soundness. For any fixed time horizon, a
separation between the calibrated product distribution \(\cD_q^T\) and
the miscalibrated product distribution \({\cD'_q}^T\) already implies
that, at the same horizon, a strategic predictor can strictly improve
over truthful reporting.

\section{Quantitative Lower Bounds from Impossibility}
\label{sec:quantitative-bounds}

In this section, we quantify the reduction behind
\Cref{thm:hybrid-impossibility}. For each \(q\in(0,1)\), let
\(\Delta_T(q)\) denote the offline \emph{distinguishing gap} between the
calibrated product distribution \(\cD_q^T\) and the miscalibrated product
distribution \((\cD'_q)^T\). Let \(\Adv^{\mathrm{prod}}_T(\Cal)\) denote
the largest expected \emph{gain} available to a strategic predictor under
a product distribution. The hybrid argument converts any positive
offline gap into a positive incentive to deviate from truthful reporting.
Quantitatively, the reduction loses a factor from averaging over the
\(T\) hybrid steps and another from the probability of reaching the
history on which the strategic predictor deviates. This gives a lower
bound of order \(\frac{(1-q)^{T-1}}{qT}\Delta_T(q)\), which is
exponentially small in the horizon.

\subsection{Hybrid Lower Bound}

For a calibration measure $\Cal_T$, define the product-distribution
strategic advantage
\[
\Adv^{\mathrm{prod}}_T(\Cal)
:=
\sup_{p^\ast,S}
\left(
\E[\Cal_T(p^\ast,Y)]-\E[\Cal_T(S(Y),Y)]
\right)_{+},
\]
where the supremum is over all product distributions
$Y_t\sim\Ber(p_t^\ast)$ and all prediction strategies $S$. Thus
$\Adv^{\mathrm{prod}}_T(\Cal)=0$ exactly means that no strategy can
improve on truthful reporting under a product distribution, while
$\Adv^{\mathrm{prod}}_T(\Cal)\le\beta_T$ is additive
$\beta_T$-truthfulness.

For $q\in(0,1)$, define
\[
\cD_q=(1-q)\delta_{(q,0)}+q\delta_{(q,1)}
\qquad\text{and}\qquad
\cD'_q=(1-q)\delta_{(q,0)}+q\delta_{(1,1)}.
\]
The corresponding distinguishing gap is
\[
\Delta_T(q)
:=
\E_{Z\sim(\cD'_q)^T}[\Cal_T(Z)]
-
\E_{Z\sim\cD_q^T}[\Cal_T(Z)].
\]

\begin{theorem}[Quantitative hybrid lower bound]
\label{thm:quantitative-hybrid-impossibility}
For every $q\in(0,1)$,
\[
\Adv^{\mathrm{prod}}_T(\Cal)
\ge
\frac{(1-q)^{T-1}}{qT}\Delta_T(q).
\]
\end{theorem}

\begin{proof}
If $\Delta_T(q)\le0$ the claim is immediate, so assume
$\Delta_T(q)>0$. Let $\cH_m$ be the hybrid distributions from the proof
of \Cref{thm:hybrid-impossibility}. By telescoping,
\[
\Delta_T(q)
=
\sum_{m=1}^T
\left(
\E_{Z\sim\cH_m}[\Cal_T(Z)]
-
\E_{Z\sim\cH_{m-1}}[\Cal_T(Z)]
\right).
\]
Therefore for some $m\in[T]$,
\[
\E_{Z\sim\cH_m}[\Cal_T(Z)]
-
\E_{Z\sim\cH_{m-1}}[\Cal_T(Z)]
\ge
\frac{\Delta_T(q)}{T}.
\]

Let
\[
Z_{<m}=((P_1,Y_1),\ldots,(P_{m-1},Y_{m-1}))
\]
denote the common prefix random vector of $\cH_{m-1}$ and $\cH_m$. Its
support is contained in
\[
\{(q,0),(1,1)\}^{m-1}.
\]

For a prefix
$z_{<m}=((p_1,y_1),\ldots,(p_{m-1},y_{m-1}))$ and
$a\in\{q,1\}$, define
\[
\Phi_a(z_{<m})
:=
\E_{Y_{m+1:T}\sim\Ber(q)^{\otimes(T-m)}}
\Big[
\Cal_T(p_1,\ldots,p_{m-1},a,q,\ldots,q;
y_1,\ldots,y_{m-1},1,Y_{m+1:T})
\Big].
\]

Then, as in \Cref{sec:product-dist-impossibility},
\begin{equation}
\label{eq:quantitative-hybrid-difference}
\E_{Z\sim\cH_m}[\Cal_T(Z)]
-
\E_{Z\sim\cH_{m-1}}[\Cal_T(Z)]
=
q\,\E_{Z_{<m}}
\left[
\Phi_1(Z_{<m})-\Phi_q(Z_{<m})
\right].
\end{equation}

The expectation on the right-hand side of
\eqref{eq:quantitative-hybrid-difference} is therefore at least
$\Delta_T(q)/(qT)$. Hence there exists a prefix
\[
\hat z_{<m}
=
((\hat p_1,\hat y_1),\ldots,(\hat p_{m-1},\hat y_{m-1}))
\]
in the support of \(Z_{<m}\) such that
\begin{equation}
\label{eq:quantitative-local-gap}
\Phi_1(\hat z_{<m})-\Phi_q(\hat z_{<m})
\ge
\frac{\Delta_T(q)}{qT}.
\end{equation}

Embed this prefix into the same product distribution as in the
qualitative proof: $p_t^\ast=\hat p_t$ for $t<m$, $p_m^\ast=1$, and
$p_t^\ast=q$ for $t>m$. Let $S$ be the strategy that reports $q$ at
time $m$ on the prefix event
\[
E=\{Y_1=\hat y_1,\ldots,Y_{m-1}=\hat y_{m-1}\}
\]
and otherwise reports truthfully. Exactly as in the proof of
\Cref{thm:hybrid-impossibility},
\[
\Pr(E)
=
\prod_{t:\,\hat p_t=q}(1-q)
=(1-q)^{|\{t<m:\hat p_t=q\}|}
\ge
(1-q)^{T-1}.
\]
Conditioned on $E$, \eqref{eq:quantitative-local-gap} is exactly
\[
\E[\Cal_T(p^\ast,Y)-\Cal_T(S(Y),Y)\mid E]
\ge
\frac{\Delta_T(q)}{qT}.
\]
On $E^c$, the two strategies coincide. Therefore
\[
\E[\Cal_T(p^\ast,Y)]-\E[\Cal_T(S(Y),Y)]
\ge
\Pr(E)\cdot\frac{\Delta_T(q)}{qT}
\ge
\frac{(1-q)^{T-1}}{qT}\Delta_T(q).
\]
Taking the supremum over product distributions and strategies proves the
theorem.
\end{proof}

For complete and sound calibration measures, the distinguishing gap for
each fixed \(q\) is eventually bounded below by a positive constant.
Choosing \(q\) as a function of a target exponential rate gives the
following consequence.

\begin{corollary}[Exponential lower bound under completeness and soundness]
\label{cor:complete-sound-quantitative-advantage}
Let \(\Cal=(\Cal_T)_{T\ge1}\) be complete and sound. For every
\(\rho\in(0,1)\), there exists a constant \(c_\rho>0\) such that, for all
sufficiently large \(T\),
\[
\Adv_T^{\mathrm{prod}}(\Cal)
\ge
c_\rho\rho^T.
\]
\end{corollary}

\begin{proof}
Fix \(\rho\in(0,1)\), and set
\[
q:=\frac{1-\rho}{2}.
\]
Then \(q\in(0,1)\) and
\[
1-q=\frac{1+\rho}{2}>\rho.
\]
The distribution \(\cD_q\) is calibrated, whereas \(\cD'_q\) is
miscalibrated. Completeness and soundness therefore imply that there are
constants \(\delta_\rho>0\) and \(T_0\) such that
\[
\Delta_T(q)\ge\delta_\rho
\]
for every \(T\ge T_0\). By
\Cref{thm:quantitative-hybrid-impossibility},
\[
\Adv_T^{\mathrm{prod}}(\Cal)
\ge
\frac{\delta_\rho(1-q)^{T-1}}{qT}.
\]
Since \((1-q)/\rho>1\), the ratio
\[
\frac{(1-q)^{T-1}}{T\rho^T}
\]
tends to infinity. In particular, this ratio is at least \(1\) for all
sufficiently large \(T\). The claim follows with
\(c_\rho=\delta_\rho/q>0\).
\end{proof}

This quantitative statement also constrains additive truthfulness. The
next corollary rewrites the lower bound as a necessary upper bound on the
distinguishing gap between the two product distributions.

\begin{corollary}[Necessary condition for additive truthfulness]
\label{cor:hybrid-necessary-condition}
If $\Cal_T$ is $\beta_T$-additively truthful for product distributions,
then for every $q\in(0,1)$,
\[
\Delta_T(q)
\le
\frac{qT}{(1-q)^{T-1}}\beta_T.
\]
\end{corollary}

\begin{proof}
Additive $\beta_T$-truthfulness implies
$\Adv^{\mathrm{prod}}_T(\Cal)\le\beta_T$. Apply
\Cref{thm:quantitative-hybrid-impossibility} and rearrange.
\end{proof}

\section{Approximately Truthful Reductions}
\label{sec:approximate-truthfulness}

\Cref{sec:product-dist-impossibility} shows that exact truthfulness is
incompatible with completeness and soundness, even for product
distributions, and \Cref{sec:weak-soundness-appendix} extends this to
weak soundness requirements.  This impossibility comes from exponentially rare prefix realizations, where a truthful predictor looks miscalibrated and is thus incentivized to correct historical realizations by strategizing in the future. 
Approximate truthfulness, however, relaxes this truthful requirement on every historical prefix realization. Approximate truthfulness only requires the truthful expected error to approximate the expected error of an adaptive strategy, where both expectations are calculated ex-ante. 

In this section, we reduce approximate truthfulness to a concentration
property of a sound and complete base calibration measure.  The reduction
has two stages.  First, we subtract a sampling-error threshold from the base
measure.  If the error remaining under truthful reporting is uniformly
small, this thresholded measure is additively approximately truthful.  We
then add a vanishing squared-loss term to the additive construction.  The
conditional squared-loss decomposition supplies a variance-scale lower
bound that turns the additive comparison into a multiplicative one.

The base calibration measure need not be truthful by itself.  It is
sufficient that, under truthful reporting, the part of the error that
remains after subtracting the sampling error is small compared to the
truthful squared loss
\[
\frac1T\sum_{t=1}^T\E[(p_t^\ast-Y_t)^2]
=
\frac1T\sum_{t=1}^T\E[p_t^\ast(1-p_t^\ast)].
\]

We show that both stages preserve the statistical guarantees of the base
measure, and apply them to smooth calibration error in
\Cref{subsec:smce-application}. We also establish uniform completeness for
bounded kernel calibration error in \Cref{subsec:more-errors-kernel}.

\subsection{The Reduction}

We first explain the requirements of the base measure for the reduction to
hold. Let $\Err=(\Err_T)_{T\ge1}$ be a sequence of base calibration
measures, where
$\Err_T:[0,1]^T\times\{0,1\}^T\to[0,1]$.
Write
\[
\overline V_T
:=
\frac1T\sum_{t=1}^T\E[p_t^\ast(1-p_t^\ast)]
=
\E\left[\frac1T\sum_{t=1}^T(p_t^\ast-Y_t)^2\right].
\]
This $\overline V_T$ is the average intrinsic variance of the
outcome-generating process.

We require uniform completeness of the error. When the predictor is truthful, uniform completeness gives a quantification of the sampling error in the calibration measure. Precisely, there exists $\tau_T$ accounting for the sampling error, and $\eta_T$ accounting for the tail, such that, after subtracting the sampling error $\tau_T$, the tail can be controlled by multiples of the variance $\eta_T\overline V_T$.

\begin{definition}[Uniform completeness]
\label{def:uniform-completeness}
The error $\Err$ satisfies \emph{uniform completeness} if
there exist sequences $(\tau_T)_{T\ge1}$ and $(\eta_T)_{T\ge1}$ with
\[
\tau_T>0,
\qquad
\eta_T\ge0,
\qquad
\tau_T\to 0,
\qquad
\eta_T\to 0,
\]
such that the following holds for every $T$ and every joint distribution
of $Y_1,\ldots,Y_T$. If
$p_t^\ast=\E[Y_t\mid Y_{<t}]$ denotes the truthful report, then
\begin{equation}
\label{eq:uniform-completeness}
\E\left[
\bigl(\Err_T(p^\ast,Y)-\tau_T\bigr)_+
\right]
\le
\eta_T\overline V_T.
\end{equation}
The sequences $\tau_T$ and $\eta_T$ are independent of the joint
distribution of $Y$.
\end{definition}

We note that uniform completeness is not directly implied by the completeness of a calibration measure. 

\subsubsection{Additively Approximate Truthfulness}

To achieve additively approximate truthfulness, the reduction subtracts the sampling-error threshold and
truncates at zero.  Define
\begin{equation}
    \label{eq: additive reduction}
\Cal_T^{\mathrm{add}}(r,y)
:=
\bigl(\Err_T(r,y)-\tau_T\bigr)_+.
\end{equation}

Uniform completeness controls the expected value of this measure
under truthful reporting. Since the measure is non-negative, additive approximate truthfulness follows immediately from the tail bound.

\begin{theorem}[Additive-truthfulness reduction]
\label{thm:generic-additive-truthfulness}
Suppose $\Err=(\Err_T)_{T\ge1}$ satisfies uniform completeness with
$\tau_T$ and $\eta_T$ as in \Cref{def:uniform-completeness}.  Then for every
joint distribution of $Y$ and every prediction strategy $S$,
\[
\E\left[\Cal_T^{\mathrm{add}}(p^\ast,Y)\right]
\le
\eta_T\overline V_T
\le
\E\left[\Cal_T^{\mathrm{add}}(S(Y),Y)\right]
+\eta_T\overline V_T.
\]
In particular, $\Cal_T^{\mathrm{add}}$ is
$\eta_T/4$-additively truthful for arbitrary joint distributions.
\end{theorem}

\begin{proof}
By \eqref{eq:uniform-completeness},
\[
\E\left[\Cal_T^{\mathrm{add}}(p^\ast,Y)\right]
\le \eta_T\overline V_T.
\]
On the other hand,
$\E[\Cal_T^{\mathrm{add}}(S(Y),Y)]\ge0$ for every prediction strategy
$S$.  Adding this non-negative quantity gives the first claim.  Finally,
$p_t^\ast(1-p_t^\ast)\le1/4$ almost surely for every $t$, so
$\overline V_T\le1/4$.  The second claim follows from
\Cref{def:additive-truthfulness}.
\end{proof}

\subsubsection{Multiplicatively Approximate Truthfulness}

The multiplicative truthfulness reduction adds a normalized squared-loss term to
the additive construction.  For $\lambda_T>0$, define
\begin{equation}
\label{eq: mult reduction}
    \Cal_T^{\mathrm{mult}}(r,y)
:=
\Cal_T^{\mathrm{add}}(r,y)
+\lambda_T\cdot\frac1T\sum_{t=1}^T(r_t-y_t)^2.
\end{equation}

The additive construction in \Cref{eq: additive reduction} controls the absolute gain from a deviation, but
does not give a relative comparison when the error is close to
zero. This squared loss provides the error needed for a relative comparison. 
To see this, every prediction strategy pays at least
$\lambda_T\overline V_T$ in expected squared loss: 
for every predictable report $r_t$,
\[
\E[(r_t-Y_t)^2\mid Y_{<t}]
=
(r_t-p_t^\ast)^2+p_t^\ast(1-p_t^\ast)
\ge p_t^\ast(1-p_t^\ast),
\]
which is of the same order as the squared loss term. 

The following theorem formalizes this second stage of the reduction.

\begin{theorem}[Multiplicative-truthfulness reduction]
\label{thm:generic-multiplicative-truthfulness}
Suppose $\Err=(\Err_T)_{T\ge1}$ satisfies uniform completeness with
$\tau_T$ and $\eta_T$ as in \Cref{def:uniform-completeness}. Then for
every $\lambda_T>0$, the measure
\[
\Cal_T^{\mathrm{mult}}(r,y)
:=
\bigl(\Err_T(r,y)-\tau_T\bigr)_+
+
\lambda_T\cdot\frac1T\sum_{t=1}^T(r_t-y_t)^2
\]
is $\left(1+\eta_T/\lambda_T\right)$-multiplicatively truthful for
arbitrary joint distributions.
\end{theorem}

\begin{proof}
Fix a joint distribution of $Y$ and write
\[
\mu=\sum_{t=1}^T\E[p_t^\ast(1-p_t^\ast)].
\]
If $\mu=0$, then $Y_t=p_t^\ast$ almost surely for every $t$, and the
assumed bound forces the truthful thresholded error and truthful squared
loss both to be zero, so the conclusion is immediate. Assume now that $\mu>0.$
By \eqref{eq:uniform-completeness} and
$\E[(p_t^\ast-Y_t)^2]=\E[p_t^\ast(1-p_t^\ast)]$,
\[
\E[\Cal_T^{\mathrm{mult}}(p^\ast,Y)]
\le
(\eta_T+\lambda_T)\frac{\mu}{T}.
\]

Now let $S$ be any prediction strategy and set $r_t=S_t(Y_{<t})$.  Since
$r_t$ is measurable with respect to the past,
\[
\E[(r_t-Y_t)^2\mid Y_{<t}]
=
(r_t-p_t^\ast)^2+p_t^\ast(1-p_t^\ast)
\ge
p_t^\ast(1-p_t^\ast).
\]
Therefore
\[
\E[\Cal_T^{\mathrm{mult}}(S(Y),Y)]
\ge
\lambda_T\cdot\frac1T\sum_{t=1}^T\E[(r_t-Y_t)^2]
\ge
\lambda_T\frac{\mu}{T}.
\]
Combining the two inequalities gives
\[
\E[\Cal_T^{\mathrm{mult}}(p^\ast,Y)]
\le
\left(1+\frac{\eta_T}{\lambda_T}\right)
\E[\Cal_T^{\mathrm{mult}}(S(Y),Y)].
\]
\end{proof}

\subsection{Preserving Soundness and Completeness}

The next result covers both stages of the reduction: setting
$\lambda_T=0$ gives the additive construction, while $\lambda_T>0$ gives
the multiplicative construction.

\begin{theorem}[Preservation of soundness and completeness]
\label{thm:preserve-completeness-soundness}
Let
\[
\Err_T:[0,1]^T\times\{0,1\}^T\to\R_{\ge0}
\]
be a sequence of base calibration measures.  Let
\(\tau_T,\lambda_T\ge0\) satisfy \(\tau_T\to0\) and
\(\lambda_T\to0\), and define
\[
\Cal_T(r,y)
:=
\bigl(\Err_T(r,y)-\tau_T\bigr)_+
+
\lambda_T\cdot\frac1T\sum_{t=1}^T(r_t-y_t)^2.
\]
If \(\Err=(\Err_T)_{T\ge1}\) is complete, then
\(\Cal=(\Cal_T)_{T\ge1}\) is complete.  If \(\Err\) is sound, then
\(\Cal\) is sound.
\end{theorem}

\begin{proof}
We first prove completeness.  Fix a calibrated one-sample distribution
\(\cD\) over \([0,1]\times\{0,1\}\), and draw
\[
Z=((R_1,Y_1),\ldots,(R_T,Y_T))\sim\cD^T.
\]
Since \(\tau_T\ge0\),
\[
\bigl(\Err_T(Z)-\tau_T\bigr)_+\le \Err_T(Z).
\]
Moreover, calibration of \(\cD\) implies
\[
\E[(R-Y)^2]
=
\E\!\left[\E[(R-Y)^2\mid R]\right]
=
\E[R(1-R)]
\le \frac14.
\]
Therefore
\[
\E_{Z\sim\cD^T}[\Cal_T(Z)]
\le
\E_{Z\sim\cD^T}[\Err_T(Z)]
+\frac{\lambda_T}{4}.
\]
If \(\Err\) is complete, the first term tends to zero for every calibrated
\(\cD\), and the second term tends to zero by assumption.  Hence
\(\Cal\) is complete.

We next prove soundness.  Fix a miscalibrated one-sample distribution
\(\cD\).  For every \(x\ge0\) and every \(\tau\ge0\),
\((x-\tau)_+\ge x-\tau\).  Since the squared-loss term is non-negative,
\[
\E_{Z\sim\cD^T}[\Cal_T(Z)]
\ge
\E_{Z\sim\cD^T}[\Err_T(Z)]-\tau_T.
\]
If \(\Err\) is sound, then
\[
\liminf_{T\to\infty}\E_{Z\sim\cD^T}[\Err_T(Z)]>0.
\]
Because \(\tau_T\to0\), the same positive lower limit is retained by
\(\Cal\).  Thus \(\Cal\) is sound.
\end{proof}

\subsection{Application to Smooth Calibration Error}
\label{subsec:smce-application}

We now apply both reductions to the normalized smooth calibration error
from \Cref{sec:preliminaries}. Its truthful error satisfies an
$\exp(-\Omega(T\tau^2))$ tail bound, proved in \Cref{sec:smCE-conc},
which gives the following uniform-completeness parameters.

\begin{lemma}[Uniform completeness of smooth calibration error]
\label{lem:smCE-uniform-completeness}
Fix $0<\alpha<1/2$. There exists $T_0=T_0(\alpha)$ such that normalized
smooth calibration error satisfies uniform completeness with
\[
\tau_T=T^{-\alpha},
\qquad
\eta_T=
\begin{cases}
2, & T<T_0,\\
\exp\!\left(-T^{1-2\alpha}/32\right), & T\ge T_0.
\end{cases}
\]
\end{lemma}

\begin{samepage}
\begin{corollary}[Approximately truthful smooth calibration measures]
\label{cor:smCE-truthfulness}
Fix $0<\alpha<1/2$, and define
\begin{align*}
\Cal_T^{\mathrm{sm,add}}(r,y)
&:=\bigl(\smCE_T(r,y)-T^{-\alpha}\bigr)_+,\\
\Cal_T^{\mathrm{sm,mult}}(r,y)
&:=\Cal_T^{\mathrm{sm,add}}(r,y)
+T^{-\alpha}\cdot\frac1T\sum_{t=1}^T(r_t-y_t)^2.
\end{align*}
Both sequences are complete and sound. For all sufficiently large $T$
and arbitrary joint outcome distributions,
\begin{enumerate}
\item $\Cal_T^{\mathrm{sm,add}}$ is
$\frac14\exp(-T^{1-2\alpha}/32)$-additively truthful;
\item $\Cal_T^{\mathrm{sm,mult}}$ is
$\bigl(1+\exp(-T^{1-2\alpha}/64)\bigr)$-multiplicatively truthful.
\end{enumerate}
In particular, for every $0<\vareps<1$, taking $\alpha=\vareps/2$ gives
a complete and sound calibration measure that is
$\bigl(1+\exp(-\Omega(T^{1-\vareps}))\bigr)$-multiplicatively truthful.
\end{corollary}
\end{samepage}

\begin{proof}
By \Cref{lem:smCE-uniform-completeness}, uniform completeness holds
with $\tau_T=T^{-\alpha}$ and
$\eta_T=\exp(-T^{1-2\alpha}/32)$ for all sufficiently large $T$.
\Cref{thm:generic-additive-truthfulness} gives the additive guarantee.
Applying \Cref{thm:generic-multiplicative-truthfulness} with
$\lambda_T=T^{-\alpha}$ gives the factor
\[
1+\frac{\eta_T}{\lambda_T}
=1+T^\alpha\exp\!\left(-\frac{T^{1-2\alpha}}{32}\right)
\le1+\exp\!\left(-\frac{T^{1-2\alpha}}{64}\right)
\]
for all sufficiently large $T$. Finally, normalized smooth calibration
error is complete and sound \citep{smooth,utc}, and
$\tau_T,\lambda_T\to0$, so \Cref{thm:preserve-completeness-soundness}
preserves both properties for both constructions.
\end{proof}

\begin{remark}
The choice of $\alpha$ is a tradeoff. Smaller values make the threshold
and the squared-loss weight vanish more slowly, but give a stronger
truthfulness guarantee. Larger values make these perturbations vanish
faster, but weaken the concentration bound.
\end{remark}

\subsection{Kernel Calibration Error}
\label{subsec:more-errors-kernel}

We next establish uniform completeness (\Cref{def:uniform-completeness})
for bounded kernel calibration error. A second-moment bound gives a
polynomially decaying $\eta_T$
(\Cref{prop:more-errors-kernel-second-moment}), while an exponential
bound makes the multiplicative truthfulness factor exponentially close
to $1$ (\Cref{prop:more-errors-kernel-exponential}).

\paragraph{Setup.}
Let \(K\) be a positive semidefinite kernel on \([0,1]\), with
reproducing kernel Hilbert space \(H\) and feature map
\(\phi(p)=K(p,\cdot)\), so that \(h(p)=\langle h,\phi(p)\rangle_H\) for
every \(h\in H\). We normalize the kernel so that
\(\sup_{p\in[0,1]}K(p,p)\le1\), or equivalently
\(\|\phi(p)\|_H\le1\) for every \(p\).

For a transcript \((r,y)\), the empirical kernel calibration error is
\begin{align}
\operatorname{kCE}_T^K(r,y)
&=\sup_{\|h\|_H\le1}\frac1T\sum_{t=1}^T(y_t-r_t)h(r_t)\notag\\
&=\left\|\frac1T\sum_{t=1}^T(y_t-r_t)\phi(r_t)\right\|_H\notag\\
&=\left(\frac1{T^2}\sum_{s,t=1}^T
 (y_s-r_s)(y_t-r_t)K(r_s,r_t)\right)^{1/2}.
\label{eq:more-errors-kernel-definition}
\end{align}
The second equality follows from the reproducing property and
Cauchy--Schwarz, and the third from expanding the squared norm. This is
the empirical version of the kernel calibration error of
\citet[Definition~8.1 and Lemma~8.2]{utc}. Each summand in the second
line has norm at most \(1/T\), so \(0\le\operatorname{kCE}_T^K\le1\).

\begin{proposition}[Second-moment bound]
\label{prop:more-errors-kernel-second-moment}
For every $T\ge1$, every joint outcome distribution, and every $\tau>0$,
\[
\E[(\operatorname{kCE}_T^K(p^\ast,Y)-\tau)_+]
\le\frac{1}{4T\tau}\overline V_T.
\]
Consequently, kernel calibration error satisfies uniform completeness
with $\tau_T=T^{-1/2}$ and $\eta_T=(4\sqrt T)^{-1}$. More generally,
any $\tau_T\to0$ with $T\tau_T\to\infty$ works with
$\eta_T=(4T\tau_T)^{-1}$.
\end{proposition}

The proof is deferred to \Cref{sec:kernel-completeness-proofs}.

\begin{proposition}[Exponential uniform completeness]
\label{prop:more-errors-kernel-exponential}
For every joint outcome distribution, $0<\tau\le1$, and
$Q:=T\tau^2\ge4$,
\begin{equation}
\E[(\operatorname{kCE}_T^K(p^\ast,Y)-\tau)_+]
\le\Pr\{\operatorname{kCE}_T^K(p^\ast,Y)\ge\tau\}
\le Qe^{-Q/4}\overline V_T.
\label{eq:more-errors-kernel-exponential-tail}
\end{equation}
In particular, whenever $Q\ge32$,
\begin{equation}
\E[(\operatorname{kCE}_T^K(p^\ast,Y)-\tau)_+]
\le\Pr\{\operatorname{kCE}_T^K(p^\ast,Y)\ge\tau\}
\le e^{-T\tau^2/8}\overline V_T.
\label{eq:more-errors-kernel-exponential-excess}
\end{equation}
\end{proposition}

Both results hold for every joint outcome distribution, with
$\overline V_T$ as defined at the start of this section. Their proofs
appear in \Cref{sec:kernel-completeness-proofs}; the exponential bound
uses a Hilbert-space exponential-moment argument that retains the
intrinsic-variance factor $\overline V_T$.

\paragraph{Consequences for uniform completeness and truthfulness.}
Fix \(0<\gamma<1/2\) and apply
\Cref{prop:more-errors-kernel-exponential} with \(\tau=T^{-\gamma}\).
Then \(Q=T^{1-2\gamma}\to\infty\), so \(Q\ge32\) for all sufficiently
large \(T\), and \eqref{eq:more-errors-kernel-exponential-excess} gives
\begin{equation}
\label{eq:more-errors-kernel-tail}
\Pr\!\left[\operatorname{kCE}_T^K(p^\ast,Y)\ge T^{-\gamma}\right]
\le \exp\!\left(-T^{1-2\gamma}/8\right)\overline V_T.
\end{equation}
The general transfer in \Cref{sec:more-error-bar-bounds} uses this tail
bound to establish uniform completeness for other measures. In
particular, kernel calibration error
satisfies uniform completeness with \(\tau_T=T^{-\gamma}\) and
\(\eta_T=\exp(-T^{1-2\gamma}/8)\) for all sufficiently large \(T\). For
the finitely many smaller horizons, take \(\tau_T=1\) and \(\eta_T=0\),
which is valid because \(\operatorname{kCE}_T^K\le1\).

Applying the multiplicative reduction
(\Cref{thm:generic-multiplicative-truthfulness}) to kernel calibration
error with \(\lambda_T=T^{-\gamma}\) therefore gives a
\((1+\eta_T/\lambda_T)\)-multiplicatively truthful measure, where for
sufficiently large \(T\)
\[
1+T^\gamma\exp\!\left(-T^{1-2\gamma}/8\right)
\le 1+\exp\!\left(-T^{1-2\gamma}/16\right).
\]
These concentration bounds establish uniform completeness. For the resulting measure to
also be sound (\Cref{thm:preserve-completeness-soundness}), the kernel
calibration error itself must be sound, which depends on the kernel.
For instance, the constant kernel $K(u,v)=1$ detects only the mean
residual and can vanish on a miscalibrated distribution.

\section{Reductions Between Calibration Measures}
\label{sec:more-error-bar-bounds}

Uniform completeness can be transferred from one calibration measure to
another whenever a pointwise comparison relates their values on every
transcript. The source measure need not be smooth calibration error or
a kernel calibration error. We first state the general reduction, then
apply it to consistent calibration measures using the kernel bound
from \Cref{subsec:more-errors-kernel}.

Throughout this section, $p_t^\ast=\E[Y_t\mid Y_{<t}]$ and
$\overline V_T=T^{-1}\sum_{t=1}^T\E[p_t^\ast(1-p_t^\ast)]$ for an
arbitrary joint distribution of binary outcomes.

\subsection{A General Pointwise Reduction}
\label{subsec:more-errors-transfer}

\begin{lemma}[Transfer by pointwise domination]
\label{lem:more-errors-transfer}
Let $\Err=(\Err_T)_{T\ge1}$ and $\Cal=(\Cal_T)_{T\ge1}$ be two
calibration measures taking values in $[0,1]$. Suppose $b_T\ge0$ tends
to zero and $\omega:[0,\infty)\to[0,\infty)$ is nondecreasing, with
$\omega(u)>0$ for $u>0$ and $\omega(u)\to0$ as $u\to0$, such that
every transcript satisfies
\begin{equation}
\label{eq:more-errors-domination}
\Cal_T(r,y)\le b_T+\omega(\Err_T(r,y)).
\end{equation}
If $\Err$ satisfies uniform completeness with parameters $s_T$ and
$\varepsilon_T$, then $\Cal$ also satisfies uniform completeness.
More precisely, for any sequence $h_T>0$ such that $h_T\to0$ and
$\varepsilon_T/h_T\to0$, valid parameters are
\begin{equation}
\label{eq:more-errors-transfer-parameters}
\tau_T=b_T+\omega(s_T+h_T),
\qquad
\eta_T=\frac{\varepsilon_T}{h_T}.
\end{equation}
Such a sequence always exists; for example,
$h_T=\sqrt{\varepsilon_T}+T^{-1}$.

If there are distribution-independent sequences
$u_T>0$ and $\rho_T\ge0$, both tending to zero, such that every joint
outcome distribution satisfies
\begin{equation}
\label{eq:more-errors-source-tail}
\Pr\{\Err_T(p^\ast,Y)>u_T\}\le\rho_T\overline V_T,
\end{equation}
then the source tail bound gives the sharper parameters
\begin{equation}
\label{eq:more-errors-tail-transfer-parameters}
\tau_T=b_T+\omega(u_T),
\qquad
\eta_T=\rho_T.
\end{equation}
\end{lemma}

\begin{proof}
For any $u>0$, the pointwise comparison implies
$\Cal_T(p^\ast,Y)\le b_T+\omega(u)$ whenever
$\Err_T(p^\ast,Y)\le u$. Since $0\le\Cal_T\le1$, it follows that
\begin{equation}
\label{eq:more-errors-transfer-indicator}
\bigl(\Cal_T(p^\ast,Y)-b_T-\omega(u)\bigr)_+
\le\ind\{\Err_T(p^\ast,Y)>u\}.
\end{equation}
Take $u=s_T+h_T$. Markov's inequality and uniform completeness of
$\Err$ give
\[
\Pr\{\Err_T(p^\ast,Y)>s_T+h_T\}
\le\frac{\E[(\Err_T(p^\ast,Y)-s_T)_+]}{h_T}
\le\frac{\varepsilon_T}{h_T}\overline V_T.
\]
Taking expectations in \eqref{eq:more-errors-transfer-indicator}
proves \eqref{eq:more-errors-transfer-parameters}. The assumptions
ensure $\tau_T>0$, $\tau_T\to0$, and $\eta_T\to0$.
Alternatively, take $u=u_T$ and apply
\eqref{eq:more-errors-source-tail} directly to obtain
\eqref{eq:more-errors-tail-transfer-parameters}.
\end{proof}

The extra threshold $h_T$ permits transfer using only the expected
excess bound in uniform completeness. A source tail bound avoids this
additional threshold and the factor $1/h_T$. In particular, if
\[
\Cal_T(r,y)\le b_T+A\Err_T(r,y)^\beta,
\qquad A>0,\quad\beta>0,
\]
then a source tail bound at level $u_T$ transfers with threshold
$b_T+A u_T^\beta$ and the same variance-normalized tail factor.
Both reductions apply to arbitrary transcript measures, including
measures that depend on the order of the reports and outcomes.

\subsection{Consistent Calibration Measures}
\label{subsec:more-errors-consistent}

Let $\Cal$ be a calibration measure over prediction-outcome pairs in $[0,1]$. Write $\Cal(\Gamma)$ for its value
on a distribution $\Gamma$ over prediction--outcome pairs. On a
length-$T$ transcript, we evaluate the same measure at the empirical
distribution:
\begin{equation}
\label{eq:more-errors-empirical}
\Cal_T(r,y):=\Cal(\widehat\Gamma_T(r,y)),
\qquad
\widehat\Gamma_T(r,y)
=\frac1T\sum_{t=1}^T\delta_{(r_t,y_t)}.
\end{equation}
We use $\Cal$ also for the resulting sequence $(\Cal_T)_{T\ge1}$ of
calibration measures on transcripts.

For a distribution $D$ over $\mathcal X\times\{0,1\}$ and a predictor
$f:\mathcal X\to[0,1]$, let $D_f$ be the distribution of $(f(X),Y)$
when $(X,Y)$ is drawn from $D$, and define
\[
d_D(f)
:=
\inf_{g:\,\E_D[Y\mid g(X)]=g(X)}
\E_D|f(X)-g(X)|.
\]
Following \citet[Definition~4.2]{utc}, the calibration
measure $\Cal$ is \emph{robustly complete} if there are constants
$a,q>0$ such that
\begin{equation}
\label{eq:more-errors-robust-completeness}
\Cal(D_f)\le a d_D(f)^q
\end{equation}
for every $D$ and $f$.  A consistent measure also satisfies a
robust-soundness inequality in the other direction.  Only robust
completeness is needed for the error-bar bound.

For a distribution $\Gamma$ over prediction--outcome pairs, define
\[
\underline d(\Gamma)
:=
\inf\left\{
\E|R-U|:
\begin{array}{l}
(R,Y)\text{ has distribution }\Gamma,\\[-2pt]
\E[Y\mid U]=U
\end{array}
\right\}.
\]
The infimum considers all distributions of $(U,R,Y)$ that satisfy the
two displayed requirements.  Also define the smooth calibration error
of $\Gamma$ by
\[
\smCE(\Gamma)
:=
\sup_{h\in\cF}
\left|\E_\Gamma[(Y-R)h(R)]\right|.
\]
In particular,
$\smCE(\widehat\Gamma_T(r,y))=\smCE_T(r,y)$.  The comparison results of
\citet[Theorems~5.5 and~7.3]{utc} give
\begin{equation}
\label{eq:more-errors-distance-comparison}
\underline d(\Gamma)
=\inf_{(D,f):\,D_f=\Gamma}d_D(f),
\qquad
\tfrac12\underline d(\Gamma)
\le\smCE(\Gamma)
\le2\underline d(\Gamma).
\end{equation}

We now specialize to the Laplace kernel $K_{\mathrm{Lap}}(u,v)=e^{-|u-v|}$.
By \citet[Theorem~8.5]{utc},
\begin{equation}
\label{eq:more-errors-laplace-comparison}
\tfrac13\smCE(\Gamma)
\le \operatorname{kCE}^{\mathrm{Lap}}(\Gamma)
\le \sqrt{\underline d(\Gamma)}.
\end{equation}
Here the population kernel error is the norm of
$\E_\Gamma[(Y-R)\phi(R)]$; its empirical version is
\eqref{eq:more-errors-kernel-definition}. In particular,
\eqref{eq:more-errors-distance-comparison} implies
\begin{equation}
\label{eq:more-errors-distance-kernel}
\underline d(\Gamma)\le 2\smCE(\Gamma)
\le 6\operatorname{kCE}^{\mathrm{Lap}}(\Gamma).
\end{equation}
This is the pointwise comparison that transfers the kernel bound to
robustly complete measures.

\begin{theorem}[Uniform completeness from robust completeness]
\label{thm:more-errors-consistency}
Suppose the calibration measure $\Cal$ takes values in $[0,1]$ and
satisfies
\eqref{eq:more-errors-robust-completeness}, and evaluate it on transcripts
as in \eqref{eq:more-errors-empirical}. Then $\Cal=(\Cal_T)_{T\ge1}$ satisfies
uniform completeness. In particular,
let
\[
A=a6^q.
\]
For every fixed $0<\gamma<1/2$, for all sufficiently large $T,$
\begin{equation}
\label{eq:more-errors-consistent-parameters}
\tau_T=A T^{-\gamma q},
\qquad
\eta_T=\exp\!\left(-T^{1-2\gamma}/8\right).
\end{equation}
Consequently, every consistent calibration measure
in the framework of \citet{utc} satisfying a polynomial robust
completeness bound has the required truthful error-bar bound under \eqref{eq:more-errors-empirical}.
\end{theorem}

\begin{proof}
Fix a distribution $\Gamma$ over prediction--outcome pairs.  Applying
\eqref{eq:more-errors-robust-completeness} to every
$(D,f)$ such that $D_f=\Gamma$, and then using
\eqref{eq:more-errors-distance-kernel}, gives
\begin{align}
\Cal(\Gamma)
&\le a\underline d(\Gamma)^q \notag\\
&\le A\operatorname{kCE}^{\mathrm{Lap}}(\Gamma)^q.
\label{eq:more-errors-consistent-domination}
\end{align}
Evaluating at $\Gamma=\widehat\Gamma_T(r,y)$ gives
\[
\Cal_T(r,y)\le A\operatorname{kCE}_T^{\mathrm{Lap}}(r,y)^q.
\]
Apply the tail-bound version of \Cref{lem:more-errors-transfer} with
$\Err_T=\operatorname{kCE}_T^{\mathrm{Lap}}$, $b_T=0$,
$\omega(u)=Au^q$, and $u_T=T^{-\gamma}$. The source tail bound
\eqref{eq:more-errors-kernel-tail} gives
$\rho_T=\exp(-T^{1-2\gamma}/8)$ for all sufficiently large $T$.
For the finitely many smaller horizons, take $\tau_T=1$ and $\eta_T=0$.
\end{proof}

The fixed-measure assumption ensures that $\Cal(\Gamma)$ is the same
for every $(D,f)$ with $D_f=\Gamma$. This permits the
infimum over all such $(D, f)$ in
\eqref{eq:more-errors-consistent-domination}. Evaluating $\Cal$ directly on
the empirical distribution then transfers this bound to every
transcript. An arbitrary estimator of $\Cal(\Gamma)$ need not satisfy the
same pointwise bound merely because it is statistically consistent.

If $\Cal$ also satisfies robust soundness
$\Cal(D_f)\ge c d_D(f)^s$ for constants $c,s>0$, then the same comparisons
give the two-sided pointwise bound
\begin{equation}
\label{eq:more-errors-two-sided}
c2^{-s}\smCE(\Gamma)^s
\le \Cal(\Gamma)
\le a2^q\smCE(\Gamma)^q.
\end{equation}
Thus robust completeness supplies the upper bound needed for uniform
completeness, while robust soundness supplies the lower bound needed to
detect miscalibration.

\subsection{Examples and Approximately Truthful Measures}
\label{subsec:more-errors-examples}

Here $\overline d$ is upper calibration distance, obtained by restricting
the calibrated prediction to a deterministic function of the original
prediction, and $\operatorname{intCE}$ is interval calibration error.
Combining \citet[Theorems~6.2 and~7.3]{utc} with
\eqref{eq:more-errors-distance-kernel} gives the comparisons in
\Cref{tab:more-errors-transfer}.
Each row has the form
$\Cal_T\le A(\operatorname{kCE}_T^{\mathrm{Lap}})^\beta$ and uses
$\eta_T=\exp(-T^{1-2\gamma}/8)$ for sufficiently large $T$.

\begin{table}[ht]
\centering
\small
\renewcommand{\arraystretch}{1.2}
\begin{tabular}{@{}lccc@{}}
\toprule
Calibration measure $\Cal_T$ & $A$ & $\beta$ & Threshold $\tau_T=A T^{-\gamma\beta}$ \\
\midrule
$\underline d(\widehat\Gamma_T)$
& $6$ & $1$ & $6T^{-\gamma}$ \\
$\overline d(\widehat\Gamma_T),\ \operatorname{intCE}(\widehat\Gamma_T)$
& $4\sqrt{6}$ & $1/2$ & $4\sqrt{6}\,T^{-\gamma/2}$ \\
\bottomrule
\end{tabular}
\caption{Pointwise reductions from Laplace kernel calibration error and
resulting uniform-completeness thresholds.}
\label{tab:more-errors-transfer}
\end{table}

For any of these measures, the general construction from
\Cref{sec:approximate-truthfulness} gives
\[
\Cal_T^{\mathrm{mult}}(r,y)
:=\bigl(\Cal_T(r,y)-\tau_T\bigr)_+
+T^{-\gamma}\cdot\frac1T\sum_{t=1}^T(r_t-y_t)^2.
\]
By \Cref{thm:generic-multiplicative-truthfulness}, for all sufficiently
large $T$ this measure is
\[
\left(1+T^\gamma e^{-T^{1-2\gamma}/8}\right)
\text{-multiplicatively truthful},
\]
and the factor is at most $1+e^{-T^{1-2\gamma}/16}$ for sufficiently
large $T$. When the base measure is complete and sound,
\Cref{thm:preserve-completeness-soundness} preserves both properties,
since $\tau_T\to0$ and $T^{-\gamma}\to0$.

\section*{AI Disclosure}
The authors used AI tools to assist in developing and editing the proofs of the peeling lemma and the comparison lemma in Appendix~\ref{sec:weak-soundness-appendix}. The transfer theorem idea was obtained after providing an AI model with a hint that the $\smCE$ concentration argument might be applicable to kernel calibration measures. The authors independently verified the correctness and originality of all content, including all references.

\section*{Acknowledgments} We thank Mingda Qiao for helpful discussions at the initial stage of this work.

\bibliographystyle{plainnat}
\bibliography{ref}
\appendix
\crefalias{section}{appendix}
\crefname{appendix}{appendix}{appendices}
\Crefname{appendix}{Appendix}{Appendices}
\section{Truthfulness and Weak Soundness}
\label{sec:weak-soundness-appendix}

This appendix shows that truthfulness is incompatible even with the weak
completeness and soundness requirements of \citet{haghtalab2024truthfulness},
which only test constant predictors. Their product completeness condition
asks a constant report equal to the Bernoulli mean to have sublinear
expected error. Their product soundness condition asks a constant report
separated from that mean to have linear expected error.

The obstruction is a single inequality. For a truthful calibration
measure and any \(0<q<p<1\), we show that
\[
\E_{Y\sim\Ber(p)^{\otimes T}}\bigl[\Cal_T(q,\ldots,q;Y)\bigr]
\le
\frac{p(1-q)}{q(1-p)}\,
\E_{Y\sim\Ber(p)^{\otimes T}}\bigl[\Cal_T(p,\ldots,p;Y)\bigr].
\]
The factor does not depend on \(T\), so small error for truthful constant
reports forces small error for separated misreports, contradicting
soundness.

\Cref{sec:comparison-proof} proves this comparison lemma by induction on
\(T\). The induction is driven by a peeling lemma, which splits off the
first round of a truthful measure: the first round behaves like a proper
binary loss, and what remains is again a non-negative truthful calibration
measure on the last \(T-1\) rounds. \Cref{sec:weak-soundness-consequences}
then derives the quantitative tradeoff between completeness and soundness
errors (\Cref{thm:constant-predictor-impossibility}), the resulting
impossibility for the conditions of \citet{haghtalab2024truthfulness}
(\Cref{cor:haghtalab-weak-impossibility}), and its extension to
correlated outcomes.

\subsection{The Comparison Lemma}
\label{sec:comparison-proof}

Our goal in this subsection is the following lemma.

\begin{lemma}[Comparison lemma]
\label{lem:constant-comparison}
Let \(\Cal_T:[0,1]^T\times\{0,1\}^T\to \R_{\ge0}\) be truthful. Then for
every \(0<q<p<1\),
\[
\E_{Y\sim\Ber(p)^{\otimes T}}
\left[
\Cal_T(q,\ldots,q;Y)
\right]
\le
\frac{p(1-q)}{q(1-p)}
\E_{Y\sim\Ber(p)^{\otimes T}}
\left[
\Cal_T(p,\ldots,p;Y)
\right].
\]
\end{lemma}

Throughout this subsection we write
\[
K:=\frac{p(1-q)}{q(1-p)},
\]
and note that \(K>1\) whenever \(0<q<p<1\).

\paragraph{Proof idea.}
The proof is by induction on \(T\). For \(T=1\), a truthful measure is
just a non-negative proper binary loss, and the claim is an elementary
inequality for such losses (\Cref{lem:one-round-comparison}). For the
inductive step, we condition on the first outcome \(b\) and split the
error of the constant misreport into three pieces:
\begin{itemize}
\item a first-round loss \(\ell(q,b)\), which is a proper binary loss and
  does not depend on the tail;
\item a non-negative offset \(\beta_b(z)\), which depends only on the tail
  outcomes \(z\);
\item a residual that measures how much worse the tail reports are than
  the perfect-foresight reports \(z\).
\end{itemize}
The one-round inequality handles the first piece. The residual is itself
a non-negative truthful calibration measure on \(T-1\) rounds, so the
induction hypothesis handles it. The offset can only help, because it
enters the left-hand side with coefficient \(1\) and the right-hand side
with coefficient \(K>1\). The peeling lemma (\Cref{lem:peeling}) provides
this decomposition.

\paragraph{Proper binary losses.}
A function \(\ell:[0,1]\times\{0,1\}\to\R\) is a \emph{proper binary
loss} if, for every \(\alpha\in[0,1]\), the report \(r=\alpha\) minimizes
the expected loss \((1-\alpha)\ell(r,0)+\alpha\ell(r,1)\) over
\(r\in[0,1]\) \citep{Mcc-56,sav-71,gneiting2007strictly}. We also say that a pair of functions \((f,g)\) on
\([0,1]\) is a proper binary loss if \(\ell(r,0)=f(r)\),
\(\ell(r,1)=g(r)\) defines one. A calibration measure with \(T=1\) is a
binary loss, and it is truthful exactly when that loss is proper.

The base case of the induction, and the first round of the inductive
step, both use the following one-round version of the comparison lemma.

\begin{lemma}[One-round comparison]
\label{lem:one-round-comparison}
Let \(\ell:[0,1]\times\{0,1\}\to\R_{\ge0}\) be a non-negative proper
binary loss. Then for every \(0<q<p<1\),
\[
(1-p)\ell(q,0)+p\ell(q,1)
\le
K\bigl((1-p)\ell(p,0)+p\ell(p,1)\bigr).
\]
\end{lemma}

\begin{proof}
Since \(q<p\), we have \(1-p\le 1-q\le (p/q)(1-q)\). Using
\(\ell\ge0\) and then propriety at \(\alpha=q\),
\begin{align*}
(1-p)\ell(q,0)+p\ell(q,1)
&\le
\frac pq\bigl((1-q)\ell(q,0)+q\ell(q,1)\bigr)\\
&\le
\frac pq\bigl((1-q)\ell(p,0)+q\ell(p,1)\bigr)
=
K(1-p)\ell(p,0)+p\ell(p,1).
\end{align*}
The last expression is at most \(K\bigl((1-p)\ell(p,0)+p\ell(p,1)\bigr)\)
because \(K\ge1\) and \(\ell\ge0\).
\end{proof}

\paragraph{Notation for the first-round decomposition.}
For \(T\ge2\), we split a transcript into its first round and its
\emph{tail}, rounds \(2,\ldots,T\), and write
\[
\Cal_T(r,s;b,z)
:=
\Cal_T\bigl((r,s),(b,z)\bigr).
\]
Here \(r\in[0,1]\) is the first report, \(b\in\{0,1\}\) is the first
outcome, \(s\in[0,1]^{T-1}\) is the vector of tail reports, and
\(z\in\{0,1\}^{T-1}\) is the vector of tail outcomes. For
\(i=1,\ldots,T-1\), write \(z_{<i}=(z_1,\ldots,z_{i-1})\).

A \emph{tail strategy} is a prediction strategy for the tail, that is, a
sequence \(\sigma=(\sigma_1,\ldots,\sigma_{T-1})\) of maps
\(\sigma_i:\{0,1\}^{i-1}\to[0,1]\). Along the tail outcomes \(z\), it
produces the tail report vector
\[
\sigma[z]
:=
\bigl(
\sigma_1(\emptyset),
\sigma_2(z_1),
\ldots,
\sigma_{T-1}(z_{<T-1})
\bigr)
\in[0,1]^{T-1}.
\]
Thus \(\sigma\) is a rule for producing reports, and \(\sigma[z]\) is the
report vector it produces when the tail outcomes are \(z\).

For a distribution \(\nu\) on \(\{0,1\}^{T-1}\) and \(Z\sim\nu\), the
\emph{truthful tail strategy} \(\pi^\nu\) reports
\[
\pi^\nu_i(z_{<i})
=
\Pr_\nu[Z_i=1\mid Z_{<i}=z_{<i}]
\]
on every history with \(\Pr_\nu[Z_{<i}=z_{<i}]>0\), and \(0\) on
histories of probability zero. Two cases are used below:
\begin{itemize}
\item if \(\nu=\Ber(p)^{\otimes(T-1)}\), then
  \(\pi^\nu[z]=(p,\ldots,p)\) for every \(z\);
\item if \(\nu=\delta_z\) is the point mass at a fixed tail outcome
  sequence \(z\), then \(\pi^{\delta_z}[z]=z\). A truthful predictor
  who knows the tail in advance simply reports it. Here and below we
  view \(z\in\{0,1\}^{T-1}\) as an element of \([0,1]^{T-1}\).
\end{itemize}

\paragraph{The peeling lemma.}
The next lemma makes the three-piece decomposition from the proof idea
precise. The key quantity is \(A_b^\nu(r)\): the expected error when the
first report is \(r\), the first outcome is \(b\), and the tail, drawn
from \(\nu\), is predicted truthfully.

\begin{lemma}[Peeling lemma]
\label{lem:peeling}
Let \(T\ge 2\) and let
\(\Cal_T:[0,1]^T\times\{0,1\}^T\to\R_{\ge 0}\) be truthful. For a
distribution \(\nu\) on \(\{0,1\}^{T-1}\) and \(b\in\{0,1\}\), define
\[
A_b^\nu(r)
:=
\E_{Z\sim\nu}
\left[
\Cal_T(r,\pi^\nu[Z];b,Z)
\right].
\]
There is a non-negative proper binary loss
\(\ell:[0,1]\times\{0,1\}\to\R_{\ge0}\) with the following properties.
\begin{enumerate}
\item[(i)] \emph{The first round is a proper loss.} For every
  distribution \(\nu\) on \(\{0,1\}^{T-1}\) and every \(b\in\{0,1\}\),
  there is a constant \(B_b(\nu)\), not depending on \(r\), such that
  \[
  A_b^\nu(r)=\ell(r,b)+B_b(\nu)
  \qquad\text{for all } r\in(0,1).
  \]
\item[(ii)] \emph{Deterministic tails.} For every
  \(z\in\{0,1\}^{T-1}\) and every \(b\in\{0,1\}\), there is a constant
  \(\beta_b(z)\ge0\) such that
  \[
  \Cal_T(r,z;b,z)=\ell(r,b)+\beta_b(z)
  \qquad\text{for all } r\in(0,1).
  \]
\item[(iii)] \emph{The residual is a truthful measure on the tail.} For
  \(b\in\{0,1\}\) and \(r\in(0,1)\), define the residual
  \[
  R_{b,r}(s;z)
  :=
  \Cal_T(r,s;b,z)-\Cal_T(r,z;b,z),
  \qquad
  s\in[0,1]^{T-1},\ z\in\{0,1\}^{T-1}.
  \]
  Then \(R_{b,r}\ge0\), and \(R_{b,r}\) is truthful as a calibration
  measure on \(T-1\) rounds: for every distribution \(\Lambda\) on
  \(\{0,1\}^{T-1}\) and every tail strategy \(\sigma\),
  \[
  \E_{Z\sim\Lambda}
  \left[
  R_{b,r}(\pi^\Lambda[Z];Z)
  \right]
  \le
  \E_{Z\sim\Lambda}
  \left[
  R_{b,r}(\sigma[Z];Z)
  \right].
  \]
\end{enumerate}
\end{lemma}

The point of part~(i) is that \(\ell\) does not depend on \(\nu\); part~(iii)
is what makes induction on \(T\) possible.

\begin{proof}
We prove the three parts in turn.

\emph{Step 1: every pair \((A_0^{\nu_0},A_1^{\nu_1})\) is a proper
binary loss.} Fix two distributions \(\nu_0,\nu_1\) on
\(\{0,1\}^{T-1}\). For \(\alpha\in[0,1]\), consider the joint
distribution under which \(Y_1\sim\Ber(\alpha)\) and, conditional on
\(Y_1=b\), the tail \((Y_2,\ldots,Y_T)\) is drawn from \(\nu_b\). The
truthful strategy reports \(\alpha\) in the first round and then follows
\(\pi^{\nu_b}\) after observing \(Y_1=b\). Compare it with the strategy
that reports some \(r\) in the first round and then also follows
\(\pi^{\nu_b}\). Its expected error is
\[
(1-\alpha)A_0^{\nu_0}(r)+\alpha A_1^{\nu_1}(r),
\]
so truthfulness says that this expression is minimized at \(r=\alpha\).
Hence \((A_0^{\nu_0},A_1^{\nu_1})\) is a proper binary loss. It is
real-valued on \([0,1]\), because \(\Cal_T\) is finite-valued and the
tail outcome space is finite.

\emph{Step 2: a common loss \(\ell\).} We use the following standard
consequence of the Savage representation of proper losses
\citep{sav-71,gneiting2007strictly}, and include a short proof for
completeness: if \((f,g)\) and \((\widetilde f,g)\) are both real-valued
proper binary losses, then \(f-\widetilde f\) is constant on \((0,1)\);
likewise, if \((f,g)\) and \((f,\widetilde g)\) are, then
\(g-\widetilde g\) is constant on \((0,1)\).

To see this, let
\(L(\alpha):=\min_r\bigl[(1-\alpha)f(r)+\alpha g(r)\bigr]
=(1-\alpha)f(\alpha)+\alpha g(\alpha)\). As a minimum of affine
functions of \(\alpha\), \(L\) is concave, and hence locally Lipschitz on
\((0,1)\). For each \(r\), the affine function
\(\alpha\mapsto(1-\alpha)f(r)+\alpha g(r)\) lies above \(L\) and touches
it at \(\alpha=r\), so its slope \(g(r)-f(r)\) is a supergradient of
\(L\) at \(r\). Hence \(L'=g-f\) wherever \(L\) is differentiable, which
is almost everywhere. Define \(\widetilde L\) from
\((\widetilde f,g)\) in the same way and let \(h:=f-\widetilde f\). Then
\((1-r)h(r)=L(r)-\widetilde L(r)\) is locally Lipschitz on \((0,1)\), and
therefore so is \(h\). Differentiating almost everywhere gives
\[
-h+(1-r)h'
=
L'-\widetilde L'
=
(g-f)-(g-\widetilde f)
=
-h,
\]
so \(h'=0\) almost everywhere, and \(h\) is constant on \((0,1)\). The
case of a common first coordinate is symmetric, using
\(L-\widetilde L=r(g-\widetilde g)\).

By Step 1, every function in \(\{A_0^\nu\}_\nu\) forms a proper binary
loss with every function in \(\{A_1^\nu\}_\nu\). Fix any reference
distribution \(\bar\nu\) and set
\(\ell(r,b):=A_b^{\bar\nu}(r)\). By the fact above,
\(A_b^\nu-\ell(\cdot,b)\) is constant on \((0,1)\) for every \(\nu\) and
\(b\), which is part~(i) for this choice of \(\ell\).

\emph{Step 3: normalizing \(\ell\) (part~(ii)).} For a deterministic
tail \(z\), we have \(\pi^{\delta_z}[z]=z\), so
\[
\Cal_T(r,z;b,z)=A_b^{\delta_z}(r)=\ell(r,b)+B_b(\delta_z).
\]
There are finitely many deterministic tails, so for each \(b\) we may
choose \(z_b^\star\) minimizing \(B_b(\delta_z)\). Replace \(\ell\) by
\[
\ell(r,b):=\Cal_T(r,z_b^\star;b,z_b^\star)
=A_b^{\delta_{z_b^\star}}(r).
\]
The new \(\ell\) is a proper binary loss by Step 1, it is non-negative
because \(\Cal_T\) is, and on \((0,1)\) each coordinate differs from the
old one by a constant. So part~(i) still holds after redefining the
constants \(B_b(\nu)\). With the new \(\ell\), for every \(z\) the
difference \(\Cal_T(r,z;b,z)-\ell(r,b)\) is constant on \((0,1)\) and
equals \(B_b(\delta_z)-B_b(\delta_{z_b^\star})\), computed with the old
constants. By the minimality of \(z_b^\star\), this constant
\(\beta_b(z)\) is non-negative, which is part~(ii).

\emph{Step 4: the residual is truthful (part~(iii)).} Fix
\(b\in\{0,1\}\), \(r\in(0,1)\), a distribution \(\Lambda\) on
\(\{0,1\}^{T-1}\), and a tail strategy \(\sigma\). Consider the joint
distribution with \(\Pr[Y_1=1]=r\), under which the tail is drawn from
\(\Lambda\) conditional on \(Y_1=b\), and from an arbitrary fixed
distribution conditional on \(Y_1=1-b\). The truthful strategy reports
\(r\) in the first round and follows \(\pi^\Lambda\) after \(Y_1=b\).
Modify it to follow \(\sigma\) after \(Y_1=b\) instead, leaving
everything else unchanged. The two strategies differ only on the event
\(\{Y_1=b\}\), so the modified strategy's expected error exceeds the
truthful one's by
\[
\Pr[Y_1=b]\cdot
\Bigl(
\E_{Z\sim\Lambda}\bigl[\Cal_T(r,\sigma[Z];b,Z)\bigr]
-
\E_{Z\sim\Lambda}\bigl[\Cal_T(r,\pi^\Lambda[Z];b,Z)\bigr]
\Bigr).
\]
Truthfulness says that this is non-negative, and \(\Pr[Y_1=b]>0\)
because \(r\in(0,1)\). The term \(\Cal_T(r,z;b,z)\) subtracted in the
definition of \(R_{b,r}\) does not depend on the reports, so the
difference of expectations above equals
\(\E_\Lambda[R_{b,r}(\sigma[Z];Z)]-\E_\Lambda[R_{b,r}(\pi^\Lambda[Z];Z)]\).
This proves that \(R_{b,r}\) is truthful.

Non-negativity is the special case \(\Lambda=\delta_z\). Given any
\(s\in[0,1]^{T-1}\), choose a tail strategy with \(\sigma[z]=s\), for
instance \(\sigma_i(z_{<i})=s_i\). Since \(\pi^{\delta_z}[z]=z\),
truthfulness of \(R_{b,r}\) gives
\(R_{b,r}(s;z)\ge R_{b,r}(z;z)=0\).
\end{proof}

\paragraph{Proof of the comparison lemma.}
We can now carry out the induction outlined above.

\begin{proof}[Proof of \Cref{lem:constant-comparison}]
We induct on \(T\), proving the claim simultaneously for all non-negative
truthful calibration measures on \(T\) rounds.

\emph{Base case.} For \(T=1\), a truthful \(\Cal_1\) is a non-negative
proper binary loss, and the claim is \Cref{lem:one-round-comparison}.

\emph{Inductive step.} Let \(T\ge2\) and assume the claim for horizon
\(T-1\). Fix \(0<q<p<1\) and let \(\nu=\Ber(p)^{\otimes(T-1)}\) be the
distribution of the tail, so that \(\pi^\nu[Z]=(p,\ldots,p)\). Let
\(\ell\), \(B_b(\nu)\), \(\beta_b\), and \(R_{b,q}\) be as in
\Cref{lem:peeling}, and fix the first outcome \(b\in\{0,1\}\).

\emph{(a) Decompose.} By the definition of \(R_{b,q}\) and part~(ii),
for every tail report vector \(s\) and tail outcome sequence \(z\),
\begin{equation}
\label{eq:three-piece-decomposition}
\Cal_T(q,s;b,z)
=
\ell(q,b)+\beta_b(z)+R_{b,q}(s;z).
\end{equation}

\emph{(b) Apply the induction hypothesis to the residual.} By
part~(iii), \(R_{b,q}\) is a non-negative truthful calibration measure on
\(T-1\) rounds, so the induction hypothesis gives
\[
\E_{Z\sim\nu}\bigl[R_{b,q}(q,\ldots,q;Z)\bigr]
\le
K\,\E_{Z\sim\nu}\bigl[R_{b,q}(p,\ldots,p;Z)\bigr].
\]
To evaluate the right-hand side, take
\(s=(p,\ldots,p)=\pi^\nu[Z]\) in \eqref{eq:three-piece-decomposition}
and average over \(Z\sim\nu\). The left-hand side becomes
\(A_b^\nu(q)=\ell(q,b)+B_b(\nu)\) by part~(i), so
\[
\E_{Z\sim\nu}\bigl[R_{b,q}(p,\ldots,p;Z)\bigr]
=
B_b(\nu)-\E_{Z\sim\nu}[\beta_b(Z)].
\]
Now take \(s=(q,\ldots,q)\) in \eqref{eq:three-piece-decomposition},
average over \(Z\sim\nu\), and combine the last two displays:
\begin{align*}
\E_{Z\sim\nu}\bigl[\Cal_T(q,q,\ldots,q;b,Z)\bigr]
&\le
\ell(q,b)+\E_{Z\sim\nu}[\beta_b(Z)]
+K\bigl(B_b(\nu)-\E_{Z\sim\nu}[\beta_b(Z)]\bigr)\\
&=
\ell(q,b)+K B_b(\nu)-(K-1)\E_{Z\sim\nu}[\beta_b(Z)]\\
&\le
\ell(q,b)+K B_b(\nu),
\end{align*}
where the last step uses \(K>1\) and \(\beta_b\ge0\).

\emph{(c) Average over the first outcome.} Under
\(\Ber(p)^{\otimes T}\), the first outcome \(Y_1\sim\Ber(p)\) is
independent of the tail, which has distribution \(\nu\). Averaging the
previous bound over \(b\sim\Ber(p)\) and then applying
\Cref{lem:one-round-comparison} to \(\ell\) gives
\begin{align*}
\E_{Y\sim\Ber(p)^{\otimes T}}\bigl[\Cal_T(q,\ldots,q;Y)\bigr]
&\le
\E_{b\sim\Ber(p)}[\ell(q,b)]+K\,\E_{b\sim\Ber(p)}[B_b(\nu)]\\
&\le
K\,\E_{b\sim\Ber(p)}\bigl[\ell(p,b)+B_b(\nu)\bigr].
\end{align*}
Finally, by part~(i) and \(\pi^\nu[Z]=(p,\ldots,p)\),
\[
\ell(p,b)+B_b(\nu)
=
A_b^\nu(p)
=
\E_{Z\sim\nu}\bigl[\Cal_T(p,p,\ldots,p;b,Z)\bigr],
\]
and averaging this over \(b\sim\Ber(p)\) gives
\(\E_{Y\sim\Ber(p)^{\otimes T}}[\Cal_T(p,\ldots,p;Y)]\). This completes
the inductive step.
\end{proof}

\subsection{Consequences for Weak Completeness and Soundness}
\label{sec:weak-soundness-consequences}

For \(p,\hat p\in[0,1]\), write
\[
M_T(\hat p,p)
:=
\E_{Y\sim\Ber(p)^{\otimes T}}
\left[
\Cal_T(\hat p,\ldots,\hat p;Y)
\right]
\]
for the expected error of the constant report \(\hat p\) when the outcomes
are i.i.d.\ \(\Ber(p)\). In this notation, \Cref{lem:constant-comparison}
says that for every \(0<q<p<1\),
\[
M_T(q,p)
\le
\frac{p(1-q)}{q(1-p)}\,M_T(p,p).
\]
Since the factor does not depend on \(T\), a truthful measure cannot give
small error to truthful constant reports while giving a fixed positive
error to every separated constant report. The following definition makes
this precise.

\begin{definition}[Weak constant-predictor completeness and soundness]
\label{def:weak-constant-soundness}
Fix parameters \(c,c_1,c_2>0\). A calibration measure \(\Cal_T\)
satisfies \((c,c_1,c_2)\)-weak constant-predictor completeness and
soundness if:
\begin{enumerate}
\item \emph{(completeness)} \(M_T(p,p)\le c_1\) for every
  \(p\in[0,1]\);
\item \emph{(soundness)} \(M_T(\hat p,p^\ast)\ge c_2\) for every
  \(\hat p,p^\ast\in[0,1]\) with \(|\hat p-p^\ast|\ge c\).
\end{enumerate}
\end{definition}

Applying the comparison lemma to a symmetric pair of reports gives a
lower bound on the completeness error of any truthful measure.

\begin{theorem}[Quantitative weak-soundness tradeoff]
\label{thm:constant-predictor-impossibility}
Fix \(c\in(0,1)\) and \(c_1,c_2>0\). If a truthful calibration measure
\(\Cal_T:[0,1]^T\times\{0,1\}^T\to\R_{\ge0}\) satisfies
\((c,c_1,c_2)\)-weak constant-predictor completeness and soundness, then
\[
c_1
\ge
\left(\frac{1-c}{1+c}\right)^2c_2.
\]
\end{theorem}

\begin{proof}
Set
\[
p^\ast=\frac{1+c}{2}
\qquad\text{and}\qquad
\hat p=\frac{1-c}{2},
\]
so that \(0<\hat p<p^\ast<1\) and \(|p^\ast-\hat p|=c\). For this pair,
the comparison factor is
\[
\frac{p^\ast(1-\hat p)}{\hat p(1-p^\ast)}
=
\left(\frac{1+c}{1-c}\right)^2.
\]
Soundness, \Cref{lem:constant-comparison}, and completeness then give
\[
c_2
\le M_T(\hat p,p^\ast)
\le
\left(\frac{1+c}{1-c}\right)^2 M_T(p^\ast,p^\ast)
\le
\left(\frac{1+c}{1-c}\right)^2 c_1.
\]
Rearranging proves the claim.
\end{proof}

In particular, one cannot have \(c_1\to0\) as \(T\to\infty\) while
\(c_2\) stays bounded away from zero.

\citet{haghtalab2024truthfulness} state their product completeness and
soundness conditions for unnormalized error, so the thresholds are
\(o(T)\) and \(\Omega(T)\) rather than constants. Because the comparison
factor is scale-free, the same argument rules them out, even without
their additional requirements on deterministic outcome sequences.

\begin{corollary}[Haghtalab et al.\ completeness and soundness]
\label{cor:haghtalab-weak-impossibility}
No sequence of truthful calibration measures can satisfy both of the
following pointwise conditions:
\begin{enumerate}
\item for every \(p\in[0,1]\), \(M_T(p,p)=o_p(T)\);
\item for every distinct \(\hat p,p^\ast\in[0,1]\),
  \(M_T(\hat p,p^\ast)=\Omega_{\hat p,p^\ast}(T)\).
\end{enumerate}
More quantitatively, truthfulness forces
\[
M_T(3/4,3/4)
\ge
\frac13 M_T(1/2,3/4)
\qquad\text{for every }T.
\]
Thus a soundness bound \(M_T(1/2,3/4)\ge aT\) forces completeness error at
least \(aT/3\) at the truthful constant report \(3/4\).
\end{corollary}

\begin{proof}
It is enough to consider the fixed pair \(p^\ast=3/4\) and
\(\hat p=1/2\), for which the comparison factor is
\(\frac{(3/4)(1/2)}{(1/2)(1/4)}=3\). \Cref{lem:constant-comparison} gives,
for every \(T\),
\[
M_T(1/2,3/4)
\le
3M_T(3/4,3/4).
\]
The first condition makes the right-hand side \(o(T)\), whereas the second
condition makes the left-hand side \(\Omega(T)\), a contradiction.
\end{proof}

\paragraph{Correlated outcomes.}
One can also impose the weak conditions on correlated outcomes, requiring
them for every joint distribution on \(\{0,1\}^T\) whose coordinates all
have mean \(p\), not only for \(\Ber(p)^{\otimes T}\). Since
\(\Ber(p)^{\otimes T}\) is one such distribution, this is a stronger
requirement, and the impossibility carries over.

\begin{definition}[Correlated weak constant-predictor soundness]
\label{def:correlated-weak-constant-soundness}
Fix \(c,c_1,c_2>0\). A calibration measure \(\Cal_T\) satisfies
\((c,c_1,c_2)\)-weak correlated constant-predictor completeness and
soundness if:
\begin{enumerate}
\item for every \(p\in[0,1]\) and every joint distribution \(\Pi\) on
  \(\{0,1\}^T\) with \(\E_\Pi[Y_t]=p\) for every \(t\in[T]\),
  \[
  \E_{Y\sim\Pi}
  \left[
  \Cal_T(p,\ldots,p;Y)
  \right]
  \le c_1;
  \]
\item for every \(\hat p,p^\ast\in[0,1]\) with \(|\hat p-p^\ast|\ge c\)
  and every joint distribution \(\Pi\) on \(\{0,1\}^T\) with
  \(\E_\Pi[Y_t]=p^\ast\) for every \(t\in[T]\),
  \[
  \E_{Y\sim\Pi}
  \left[
  \Cal_T(\hat p,\ldots,\hat p;Y)
  \right]
  \ge c_2.
  \]
\end{enumerate}
\end{definition}

\begin{corollary}[Correlated weak soundness]
Fix \(c\in(0,1)\) and \(c_1,c_2>0\). If a truthful calibration measure
satisfies \((c,c_1,c_2)\)-weak correlated constant-predictor completeness
and soundness in the sense of
\Cref{def:correlated-weak-constant-soundness}, then
\[
c_1
\ge
\left(\frac{1-c}{1+c}\right)^2c_2.
\]
\end{corollary}

\begin{proof}
Restricting the correlated condition to \(\Pi=\Ber(p)^{\otimes T}\) gives
the product condition of \Cref{def:weak-constant-soundness}. The claimed
inequality follows from \Cref{thm:constant-predictor-impossibility}.
\end{proof}

\section{Concentration for Smooth Calibration Error}
\label{sec:smCE-conc}
This appendix proves the uniform-completeness bound used in
\Cref{subsec:smce-application}. We first establish a tail estimate at an
arbitrary threshold, then specialize to the parameters in
\Cref{lem:smCE-uniform-completeness}.

\begin{lemma}[Truthful tail bound for smooth calibration error]
\label{lem:smCE-helper-lemma}
Let $\smCE_T$ be the normalized smooth calibration error from
\Cref{sec:preliminaries}.  Suppose \(p_t^\ast\) is
\(\sigma(Y_1,\ldots,Y_{t-1})\)-measurable and, conditional on the past,
\(Y_t\sim\Ber(p_t^\ast)\), and write
\(\overline V_T=\frac1T\sum_{t=1}^T\E[p_t^\ast(1-p_t^\ast)]\).
For every \(0<\tau\le4/5\) with \(Q:=T\tau^2\ge8\),
\begin{equation}
\label{eq:smCE-sharp-tail}
\Pr\!\left[
\smCE_T(p^\ast,Y)\ge \tau
\right]
\le
\frac32\,Qe^{-Q/16}\,\overline V_T.
\end{equation}
\end{lemma}
The proof is a martingale concentration argument based on the
exponential-moment method for bounded martingale differences; see, e.g.,
\citep{freedman1975tail,delapena1999general,howard2020timeuniform}.
Under truthful
reporting, $Y_t-p_t^\ast$ is a bounded martingale difference, and smooth
calibration error is the supremum of predictable Lipschitz test functions
applied to these martingale differences.
\begin{proof}
We use the notation
\[
\mathcal F_t=\sigma(Y_1,\ldots,Y_t),\qquad
X_t=Y_t-p_t^\ast,\qquad
\sigma_t^2=p_t^\ast(1-p_t^\ast),\qquad
V_T=\sum_{t=1}^T\sigma_t^2,
\]
and set $\mu:=\E[V_T]$. If $\mu=0$, then $X_t=0$ for every $t$ almost surely,
and there is nothing to prove.
We first record the exponential-moment estimate used below.
Let $h_1,\ldots,h_T$ be predictable with $|h_t|\le 1$, and put
\[
D_t:=h_tX_t,\qquad
N_j:=\sum_{t=1}^jD_t,\qquad
q_t:=h_t^2\sigma_t^2,\qquad
\nu_h:=\sum_{t=1}^T\E[q_t].
\]
Fix $0<\lambda\le 1$ and define
\[
a_\lambda:=\frac{\lambda^2e^\lambda}{2}.
\]
For every $z\in[-1,1]$, Taylor's formula gives
\[
e^{\lambda z}\le 1+\lambda z+a_\lambda z^2,
\qquad
e^{-\lambda z}\le 1-\lambda z+a_\lambda z^2.
\]
Since $\E[D_t\mid\mathcal F_{t-1}]=0$ and
$\E[D_t^2\mid\mathcal F_{t-1}]=q_t$, we get
\[
\E[e^{\pm\lambda D_t}\mid\mathcal F_{t-1}]
\le
1+a_\lambda q_t.
\]
Let
\[
G_j:=e^{\lambda N_j}+e^{-\lambda N_j}-2
=2(\cosh(\lambda N_j)-1).
\]
Conditioning on $\mathcal F_{t-1}$ and using $0\le q_t\le 1$ and
$G_{t-1}\ge 0$,
\begin{align*}
\E[G_t\mid\mathcal F_{t-1}]
&\le
(1+a_\lambda q_t)\left(e^{\lambda N_{t-1}}+e^{-\lambda N_{t-1}}\right)-2\\
&=
G_{t-1}+a_\lambda q_t(G_{t-1}+2)\\
&\le
(1+a_\lambda)G_{t-1}+2a_\lambda q_t.
\end{align*}
Iterating this recursion yields
\begin{equation}
\label{eq:fb-cosh-transform}
\E[\cosh(\lambda N_T)-1]
\le
a_\lambda e^{a_\lambda T}\nu_h.
\end{equation}

Now use the one-dimensional structure of the Lipschitz class. Define
\[
R:=\sum_{t=1}^T X_t,
\qquad
B(u):=\sum_{t=1}^T\ind[u\le p_t^\ast]X_t,\qquad u\in[0,1].
\]
Every $f\in\cF$ is absolutely continuous and admits an a.e. derivative
$g_f$ with $|g_f|\le 1$.  Therefore
\[
\sum_{t=1}^T f(p_t^\ast)X_t
=
f(0)R+\int_0^1 g_f(u)B(u)\,du.
\]
Consequently, with
\[
L:=
\sup_{f\in\cF}\left|\sum_{t=1}^T f(p_t^\ast)X_t\right|,
\]
we have
\[
L\le |R|+\int_0^1|B(u)|\,du.
\]

Apply \eqref{eq:fb-cosh-transform} with $h_t=1$ to get
\[
\E[\cosh(\lambda R)-1]\le a_\lambda e^{a_\lambda T}\mu.
\]
For each fixed $u$, apply \eqref{eq:fb-cosh-transform} with
$h_t(u)=\ind[u\le p_t^\ast]$. If
\[
\nu(u):=\E\left[\sum_{t=1}^T\ind[u\le p_t^\ast]\sigma_t^2\right],
\]
then
\[
\E[\cosh(\lambda B(u))-1]\le a_\lambda e^{a_\lambda T}\nu(u).
\]
Moreover
\[
\int_0^1\nu(u)\,du
=
\E\left[\sum_{t=1}^T p_t^\ast\sigma_t^2\right]
\le \mu.
\]

Let \(s>0\). If \(L\ge s\), then either \(|R|\ge s/2\) or
$\int_0^1|B(u)|\,du\ge s/2$. The first event is bounded by Markov's
inequality:
\[
\Pr[|R|\ge s/2]
\le
\frac{a_\lambda e^{a_\lambda T}\mu}{\cosh(\lambda s/2)-1}.
\]
For the second event, Jensen's inequality gives
\[
\cosh\left(\lambda\int_0^1|B(u)|\,du\right)-1
\le
\int_0^1\left(\cosh(\lambda |B(u)|)-1\right)\,du
=
\int_0^1\left(\cosh(\lambda B(u))-1\right)\,du.
\]
Therefore another application of Markov's inequality gives
\[
\Pr\left[\int_0^1|B(u)|\,du\ge s/2\right]
\le
\frac{a_\lambda e^{a_\lambda T}\mu}{\cosh(\lambda s/2)-1}.
\]
Combining the two estimates,
\[
\Pr[L\ge s]
\le
\frac{2a_\lambda e^{a_\lambda T}\mu}{\cosh(\lambda s/2)-1}.
\]
Finally, fix \(0<\tau\le4/5\) with \(Q=T\tau^2\ge8\), and apply this
with \(s:=T\tau\) and \(\lambda:=\tau/2\le2/5\). Since
\(\smCE_T(p^\ast,Y)=L/T\), the event \(L\ge s\) is the event
\(\smCE_T(p^\ast,Y)\ge\tau\). Because \(e^\lambda\le e^{2/5}\le3/2\),
\[
a_\lambda\le \frac{3}{16}\tau^2,
\qquad
a_\lambda T\le \frac{3}{16}Q,
\qquad
\frac{\lambda s}{2}=\frac Q4\ge2.
\]
Since $\cosh(x)-1\ge e^x/4$ for $x\ge2$, it follows that
\[
\Pr\!\left[\smCE_T(p^\ast,Y)\ge\tau\right]
\le
2\cdot\frac{3\tau^2}{16}e^{3Q/16}\mu\cdot4e^{-Q/4}
=
\frac32\tau^2\mu\,e^{-Q/16}
=
\frac32\,Qe^{-Q/16}\,\overline V_T,
\]
using \(\tau^2\mu=T\tau^2\overline V_T=Q\overline V_T\). This proves the claim.
\end{proof}

\begin{proof}[Proof of \Cref{lem:smCE-uniform-completeness}]
Fix $0<\alpha<1/2$ and let $\tau_T=T^{-\alpha}$. Choose $T_0=T_0(\alpha)$
so that $\tau_T\le4/5$ and $Q:=T\tau_T^2=T^{1-2\alpha}\ge200$ for every
$T\ge T_0$. For such $Q$, we have
$\frac32Qe^{-Q/16}\le e^{-Q/32}$.

For any joint outcome distribution, the truthful reports
$p_t^\ast=\E[Y_t\mid Y_{<t}]$ satisfy the hypotheses of
\Cref{lem:smCE-helper-lemma}. Since $0\le\smCE_T\le1$, that lemma gives,
for every $T\ge T_0$,
\[
\E\left[\bigl(\smCE_T(p^\ast,Y)-\tau_T\bigr)_+\right]
\le\Pr\!\left[\smCE_T(p^\ast,Y)\ge\tau_T\right]
\le\frac32Qe^{-Q/16}\overline V_T
\le\exp\!\left(-\frac{T^{1-2\alpha}}{32}\right)\overline V_T.
\]

It remains to control the finitely many horizons $T<T_0$.  Since every
$f\in\cF$ takes values in $[-1,1]$,
\[
\smCE_T(p^\ast,Y)
\le
\frac1T\sum_{t=1}^T|Y_t-p_t^\ast|.
\]
Moreover,
\[
\E\left[|Y_t-p_t^\ast|\mid Y_{<t}\right]
=
2p_t^\ast(1-p_t^\ast).
\]
Therefore
\[
\E\left[
\bigl(\smCE_T(p^\ast,Y)-T^{-\alpha}\bigr)_+
\right]
\le
\E[\smCE_T(p^\ast,Y)]
\le
2\overline V_T.
\]
Thus $\tau_T=T^{-\alpha}$ and the piecewise sequence $\eta_T$ in
\Cref{lem:smCE-uniform-completeness} satisfy
\eqref{eq:uniform-completeness} for every $T$ and every joint distribution
of $Y$.  Both sequences converge to zero, which completes the proof.
\end{proof}

\section{Uniform Completeness of Kernel Calibration Error}
\label{sec:kernel-completeness-proofs}

This appendix proves
\Cref{prop:more-errors-kernel-second-moment,prop:more-errors-kernel-exponential}.
Fix a joint distribution of binary outcomes and use the kernel and
feature map from \Cref{subsec:more-errors-kernel}. Write
\[
\mathcal F_t=\sigma(Y_1,\ldots,Y_t),\qquad
p_t^\ast=\E[Y_t\mid\mathcal F_{t-1}],\qquad
\overline V_T=\frac1T\sum_{t=1}^T\E[p_t^\ast(1-p_t^\ast)].
\]

Under truthful reporting, the kernel calibration error is the norm of an
\(H\)-valued martingale. Define
\[
D_t:=(Y_t-p_t^\ast)\phi(p_t^\ast),
\qquad
M_t:=\sum_{s=1}^tD_s,
\qquad
M_0:=0,
\]
so that \(\operatorname{kCE}_T^K(p^\ast,Y)=\|M_T\|_H/T\). Since
\(p_t^\ast\) is \(\mathcal F_{t-1}\)-measurable and
\(\E[Y_t-p_t^\ast\mid\mathcal F_{t-1}]=0\), the increments satisfy
\(\E[D_t\mid\mathcal F_{t-1}]=0\). They are bounded,
\(\|D_t\|_H\le1\), and their conditional second moments are
\[
q_t:=\E\bigl[\|D_t\|_H^2\mid\mathcal F_{t-1}\bigr]
=p_t^\ast(1-p_t^\ast)K(p_t^\ast,p_t^\ast)
\le p_t^\ast(1-p_t^\ast).
\]
Hence, writing \(\mu:=T\overline V_T\) for the total intrinsic variance,
\(\sum_{t=1}^T\E[q_t]\le\mu\). Both proofs use only these properties of \(D_t\).

\begin{proof}[Proof of \Cref{prop:more-errors-kernel-second-moment}]
For \(s<t\), conditioning on \(\mathcal F_{t-1}\) gives
\[
\E\langle D_s,D_t\rangle_H
=\E\langle D_s,\E[D_t\mid\mathcal F_{t-1}]\rangle_H=0,
\]
so the
increments are orthogonal and
\[
\E\bigl[\operatorname{kCE}_T^K(p^\ast,Y)^2\bigr]
=\frac1{T^2}\E\|M_T\|_H^2
=\frac1{T^2}\sum_{t=1}^T\E[q_t]
\le\frac{\overline V_T}{T}.
\]
For \(x\ge0\), we have \((x-\tau)_+\le x^2/(4\tau)\), because
\(x^2/(4\tau)-(x-\tau)=(x-2\tau)^2/(4\tau)\ge0\). Applying this with
\(x=\operatorname{kCE}_T^K(p^\ast,Y)\) and taking expectations proves the
claim.
\end{proof}

\begin{proof}[Proof of \Cref{prop:more-errors-kernel-exponential}]
\emph{Step 1: a curvature bound.} For \(\theta>0\), let
\(F_\theta(x):=\cosh(\theta\|x\|_H)\) for \(x\in H\). Since \(\cosh\) is
even, \(F_\theta\) is a power series in \(\|x\|_H^2\), so it is twice
differentiable, including at \(x=0\). We claim that
\begin{equation}
\mathrm D^2F_\theta(x)[u,u]
\le\theta^2F_\theta(x)\|u\|_H^2
\qquad\text{for all }x,u\in H.
\label{eq:more-errors-kernel-hessian}
\end{equation}
At \(x=0\) the Hessian is \(\theta^2I\) and \(F_\theta(0)=1\). For
\(r=\|x\|_H>0\), the Hessian has eigenvalue \(\theta^2\cosh(\theta r)\)
in the radial direction and \(\theta\sinh(\theta r)/r\) in every
tangential direction. The latter does not exceed the former because
\(\sinh z\le z\cosh z\) for \(z\ge0\).

\emph{Step 2: a one-step inequality.} Let \(\|d\|_H\le1\). For
\(0\le s\le1\), we have \(\|x+sd\|_H\le\|x\|_H+1\), and
\(\cosh(a+b)\le e^b\cosh a\) for \(a,b\ge0\). Hence
\(F_\theta(x+sd)\le e^\theta F_\theta(x)\). Taylor's theorem with
integral remainder and \eqref{eq:more-errors-kernel-hessian} then give
\[
F_\theta(x+d)
\le F_\theta(x)+\mathrm DF_\theta(x)[d]
   +a_\theta F_\theta(x)\|d\|_H^2,
\qquad a_\theta:=\frac{\theta^2e^\theta}{2}.
\]

\emph{Step 3: iterating over time.} Apply this with \(x=M_{t-1}\) and
\(d=D_t\), and take conditional expectations given \(\mathcal F_{t-1}\).
The linear term vanishes because \(\E[D_t\mid\mathcal F_{t-1}]=0\).
Writing \(G_t:=F_\theta(M_t)-1\ge0\), we get
\begin{align*}
\E[G_t\mid\mathcal F_{t-1}]
&\le G_{t-1}+a_\theta q_t(G_{t-1}+1)\\
&\le(1+a_\theta)G_{t-1}+a_\theta q_t,
\end{align*}
using \(q_t\le1\). Take expectations and unroll from \(G_0=0\):
\(\E[G_T]\le a_\theta\sum_{t=1}^T(1+a_\theta)^{T-t}\E[q_t]\). Since
\(\sum_t\E[q_t]\le\mu\), this gives
\begin{equation}
\E[\cosh(\theta\|M_T\|_H)-1]
\le a_\theta e^{a_\theta T}\mu.
\label{eq:more-errors-kernel-cosh-moment}
\end{equation}

\emph{Step 4: choosing \(\theta\) and applying Markov's inequality.}
Choose \(\theta=\tau/2\). Since \(\tau\le1\),
\(a_\theta=\tau^2e^{\tau/2}/8\le\tau^2/4\), and hence
\(a_\theta T\le Q/4\). The event
\(\operatorname{kCE}_T^K(p^\ast,Y)\ge\tau\) is the event
\(\|M_T\|_H\ge T\tau\), and on it
\[
\cosh(\theta\|M_T\|_H)-1
\ge\cosh(Q/2)-1
\ge\tfrac14e^{Q/2},
\]
where the last step uses \(Q/2\ge2\ge\log4\). Markov's inequality
applied to \eqref{eq:more-errors-kernel-cosh-moment} now gives
\[
\Pr\{\operatorname{kCE}_T^K(p^\ast,Y)\ge\tau\}
\le\frac{a_\theta e^{a_\theta T}\mu}{\cosh(Q/2)-1}
\le\frac{\tau^2}{4}e^{Q/4}\mu\cdot4e^{-Q/2}
=Qe^{-Q/4}\overline V_T,
\]
using \(\tau^2\mu=T\tau^2\overline V_T=Q\overline V_T\). This proves
\eqref{eq:more-errors-kernel-exponential-tail}. Finally,
\(\log Q\le Q/8\) for \(Q\ge32\), so \(Qe^{-Q/4}\le e^{-Q/8}\), which
proves \eqref{eq:more-errors-kernel-exponential-excess}.
\end{proof}

\end{document}